\def\final{1}
\def\sicompversion{0}

\ifnum\sicompversion=1
\documentclass[final,legno]{siamltex}
\else
\documentclass[11pt]{article}
\fi

\usepackage{algorithm, algorithmic}
\ifnum\sicompversion=1
	\usepackage{amsmath,amssymb}
\else
	\usepackage{amsmath, amssymb, amsthm}
	\usepackage{fullpage}
	\usepackage{color}
	\definecolor{DarkGreen}{rgb}{0.1,0.5,0.1}
	\definecolor{DarkRed}{rgb}{0.5,0.1,0.1}
	\definecolor{DarkBlue}{rgb}{0.1,0.1,0.5}
	\usepackage[pdftex]{hyperref}
	\hypersetup{
	    unicode=false,          
	    pdftoolbar=true,        
	    pdfmenubar=true,        
	    pdffitwindow=false,      
	    pdftitle={New Hardness Results for Differentially Private Sanitizers},    
	    pdfauthor={Jonathan Ullman}
	    pdfsubject={},   
	    pdfnewwindow=true,      
	    pdfkeywords={keywords}, 
	    colorlinks=true,       
	    linkcolor=DarkRed,          
	    citecolor=DarkGreen,        
	    filecolor=DarkRed,      
	    urlcolor=DarkBlue,          
	}
\fi
\usepackage{verbatim}

\ifnum\final=0
\newcommand{\mynote}[1]{\marginpar{\tiny\sf #1}}
\else
\newcommand{\mynote}[1]{}
\fi
\newcommand{\jnote}[1]{\mynote{Jon: {#1}}}

\newcommand\N{\mathbb{N}}
\newcommand\R{\mathbb{R}}
\newcommand\cA{\mathcal{A}}
\newcommand\cE{\mathcal{E}}

\newcommand\cQ{\mathcal{Q}}
\newcommand\cR{\mathcal{R}}

\newcommand\poly{\mathrm{poly}}
\newcommand\bits{\{0,1\}}
\newcommand{\bitslen}[1]{\bits^{#1}}

\newcommand{\getsr}{\gets_{\mbox{\tiny R}}}
\newcommand\set[1]{\left\{#1\right\}} 
\newcommand\eps{\varepsilon}

\newcommand{\from}{\colon}

\renewcommand{\time}{T}

\newcommand\san{\mathcal{M}}
\newcommand{\db}{D}
\newcommand{\cols}{d}
\newcommand{\rows}{n}
\newcommand{\dbsetsize}[2]{(\bitslen{#1})^{#2}}
\newcommand{\dbset}{\dbsetsize{\cols}{\rows}}
\newcommand{\row}{x}
\newcommand{\rowi}[1]{\row^{(#1)}}
\newcommand{\rowseq}{\rowi{1}, \dots, \rowi{\rows}}
\newcommand{\univ}{\bitslen{\cols}}

\newcommand{\adv}{\cA}
\newcommand{\gadv}[1]{\adv_{\mathsf{#1}}}

\newcommand{\query}{q}
\newcommand{\queryj}[1]{\query_{#1}}
\newcommand{\queryset}{\cQ}
\newcommand{\querysetd}[1]{\queryset^{(#1)}}

\newcommand{\queries}{k}
\newcommand{\answer}{a}
\newcommand{\answerj}[1]{a_{#1}}
\newcommand{\queryseq}{\queryj{1}, \dots, \queryj{\queries}}
\newcommand{\answerseq}{\answerj{1}, \dots, \answerj{\queries}}
\newcommand{\ttanswerseq}{\answerj{1}, \dots, \answerj{\tracequeries}}

\newcommand{\acc}{\alpha}
\newcommand{\accfail}{\beta}

\newcommand{\decbit}{\widehat{b}}
\newcommand{\decbitj}[1]{\decbit_{#1}}
\newcommand{\decbitseq}{\decbitj{1}, \dots, \decbitj{\tracequeries}}
\newcommand{\param}{\kappa}
\newcommand{\users}{\rows}
\newcommand{\gen}{\mathsf{Gen}}
\newcommand{\ttgen}{\gen_{\mathsf{TT}}}
\newcommand{\enc}{\mathsf{Enc}}
\newcommand{\ttenc}{\enc_{\mathsf{TT}}}
\newcommand{\dec}{\mathsf{Dec}}
\newcommand{\ttdec}{\dec_{\mathsf{TT}}}
\newcommand{\malenc}{\mathsf{TrEnc_{TT}}}

\newcommand{\security}{\eps}
\newcommand{\encsecurity}{\security_{\mathsf{Enc}}}

\newcommand{\trace}{\mathsf{Trace_{TT}}}
\newcommand{\tracep}{\mathsf{Trace_{TT}^{\pirate}}}

\newcommand{\userkey}{sk}
\newcommand{\userkeyi}[1]{\userkey^{(#1)}}
\newcommand{\userkeys}{\userkeyi{1}, \dots, \userkeyi{\users}}
\newcommand{\userkeyvec}{\vec{\userkey}}

\newcommand{\ctext}{c}
\newcommand{\ctexti}[1]{\ctext^{(#1)}}
\newcommand{\ctextj}[1]{\ctext_{#1}}
\newcommand{\ctextij}[2]{\ctext_{#2}^{(#1)}}
\newcommand{\ctextseq}{\ctextj{1}, \dots, \ctextj{\tracequeries}}
\newcommand{\ctexts}{\mathcal{C}}
\newcommand{\pirate}{\mathcal{P}}
\newcommand{\tracequeries}{\queries_{\mathsf{TT}}}
\newcommand{\encqueries}{\queries_{\mathsf{Enc}}}

\newcommand{\encadv}{\gadv{Enc}}

\newcommand{\enctime}{\time_{\mathsf{Enc}}}

\renewcommand{\stretch}{s_{\mathsf{PRG}}}

\newcommand{\prg}{\mathsf{G}}
\newcommand{\prgsecurity}{\security_{\mathsf{PRG}}}

\newcommand{\prgadv}{\gadv{PRG}}
\newcommand{\decfns}{\cQ}
\newcommand{\decfn}{q}

\newcommand{\word}{w}
\newcommand{\wordi}[1]{\word^{(#1)}}

\newcommand{\wordij}[2]{\word_{#2}^{(#1)}}
\newcommand{\fword}{w'}
\newcommand{\fwordj}[1]{\fword_{#1}}
\newcommand{\wordbook}{W}
\newcommand{\compbook}{\widetilde{W}}

\newcommand{\feasible}{F}
\newcommand{\fpgen}{\gen_{\mathsf{FP}}}
\newcommand{\fptrace}{\mathsf{Trace_{FP}}}
\newcommand{\fpsecurity}{\security_{\mathsf{FP}}}
\newcommand{\wordlength}{\ell_{\mathsf{FP}}}
\newcommand{\fpadv}{\gadv{FP}}

\newcommand{\schemes}[1]{\Pi_{\mathsf{#1}}}
\newcommand{\ttscheme}{\schemes{TT}}
\newcommand{\encscheme}{\schemes{Enc}}
\renewcommand{\time}{\mathrm{T}}

\newcommand{\enckey}{\overline{\userkey}}
\newcommand{\enckeyi}[1]{\enckey^{(#1)}}
\newcommand{\fpscheme}{\schemes{FP}}
\newcommand{\crit}{\mathrm{Crit}}

\renewcommand{\tilde}{\widetilde}

\DeclareMathOperator*{\Probability}{\mathrm{Pr}}
\newcommand{\prob}[1]{\mathrm{Pr}\left[#1\right]}
\newcommand{\Prob}[2]{\Probability_{#1}\left[#2\right]}

\ifnum\sicompversion=1
	\newtheorem{claim}[theorem]{Claim}
\else
	\newtheorem{theorem}{Theorem}[section]
	\newtheorem{lemma}[theorem]{Lemma}
	
	\newtheorem{claim}[theorem]{Claim}
	
	\newtheorem{corollary}[theorem]{Corollary}

	\theoremstyle{definition}
	\newtheorem{definition}[theorem]{Definition}
\fi

\ifnum\sicompversion=1
	\title{Answering \lowercase{$n^{2+o(1)}$} Counting Queries \\ with Differential Privacy is Hard\thanks{A 10 page extended abstract of this work appeared in the 45th Annual Symposium on Theory of Computing.}}
	\author{Jonathan Ullman\thanks{School of Engineering and Applied Sciences.  Harvard University, Cambridge, MA.  This work was supposed by NSF grants CNS-0831289, CNS-1237135, and a gift from Google, Inc.  Email: jullman@seas.harvard.edu}}
\else
	\title{Answering $n^{2 + o(1)}$ Counting Queries \\ with Differential Privacy is Hard \ifnum\final=0 \\
	{\small \sc Working Draft: Please Do Not Distribute}\fi}
	\author{Jonathan Ullman\thanks{Supported by NSF grants CNS-0831289, CNS-1237235, and a gift from Google, Inc.} \\ \\ School of Engineering and Applied Sciences \\ Harvard University, Cambridge, MA \\ \href{mailto:jullman@seas.harvard.edu}{jullman@seas.harvard.edu}}
\fi

\begin{document}

\maketitle

\begin{abstract}
A central problem in differentially private data analysis is how to design efficient algorithms capable of answering large numbers of \emph{counting queries} on a sensitive database.  Counting queries are of the form ``What fraction of individual records in the database satisfy the property $q$?''  We prove that if one-way functions exist, then there is no algorithm that takes as input a database $\db \in \dbset$, and $k = \tilde{\Theta}(n^2)$ arbitrary efficiently computable counting queries, runs in time $\poly(d, n)$, and returns an approximate answer to each query, while satisfying differential privacy.  We also consider the complexity of answering ``simple'' counting queries, and make some progress in this direction by showing that the above result holds even when we require that the queries are computable by constant-depth ($AC^0$) circuits.

Our result is almost tight because it is known that $\tilde{\Omega}(n^2)$ counting queries can be answered efficiently while satisfying differential privacy.  Moreover, many more than $n^2$ queries (even exponential in $n$) can be answered in exponential time.

We prove our results by extending the connection between differentially private query release and cryptographic traitor-tracing schemes to the setting where the queries are given to the sanitizer as input, and by constructing a traitor-tracing scheme that is secure in this setting.
\end{abstract}

\ifnum\sicompversion=1


\pagestyle{myheadings}
\thispagestyle{plain}
\markboth{Answering Many Counting Queries Privately is Hard}{J. Ullman}
\fi

\section{Introduction}
Consider a database $\db \in (\bits^d)^n$, in which each of the $n$ rows corresponds to an individual's record, and each record consists of $d$ binary attributes.  The goal of privacy-preserving data analysis is to enable rich statistical analyses on the database while protecting the privacy of the individuals.  It is especially desirable to achieve \emph{differential privacy}~\cite{DworkMcNiSm06}, which guarantees that no individual's data has a significant influence on the information released about the database.

Some of the most basic statistics on a database are \emph{counting queries}, which are queries of the form, ``What fraction of individual records in $\db$ satisfy some property $q$?''  In particular we would like to construct differentially private \textbf{sanitizers} that, given a database $D$ and $k$ counting queries $q_1, \dots, q_k$ from a family $\cQ$, outputs an approximate answer to each of the queries.  
We would like the number of queries, $k$, to be as large as possible, and the set of feasible queries, $\cQ$, to be as general as possible.  Ideally, $\cQ$, would contain all counting queries.\footnote{It may require super-polynomial time just to evaluate an arbitrary counting query, which would rule out efficiency for reasons that have nothing to do with privacy.  For this discussion, we restrict attention to queries that are efficiently computable, so are not the bottleneck in the computation.}  Moreover, we would like the algorithm to run as efficiently as possible.

Some of the earliest work in differential privacy~\cite{DworkMcNiSm06} gave an efficient sanitizer---the so-called \emph{Laplace Mechanism}.  The Laplace Mechanism answers any set of $k$ arbitrary efficiently computable counting queries by perturbing the answers with appropriately calibrated random noise, providing good accuracy (say, within $\pm .01$ of the true answer) as long as $k \lesssim n^2$.

The ability to approximately answer $n^2$ counting queries is quite powerful, especially in settings where data is abundant and $n$ is large.  However, being limited to $n^2$ queries can be restrictive in settings where data is expensive or otherwise difficult to acquire, and $n$ is small.  It can also be restrictive when the budget of queries is shared between multiple analysts.  Fortunately, a remarkable result of Blum et al.~\cite{BlumLiRo08} (with subsequent developments in~\cite{DworkNaReRoVa09,DworkRoVa10, RothRo10, HardtRo10, GuptaRoUl12, HardtLiMc12}), showed that differentially private algorithms are not limited to $n^2$ queries.  They showed how to approximately answer arbitrary counting queries even when $k$ is \emph{exponentially larger} than $n$.  Unfortunately, their algorithm, and all subsequent algorithms capable of answering more than $n^2$ arbitrary counting queries, run in time (at least) $\poly(2^d, n, k)$.

The result of Blum et al., raises the exciting possibility of an \emph{efficient} algorithm that can privately compute approximate answers to large numbers of counting queries.  Unfortunately, Dwork et al.~\cite{DworkNaReRoVa09} gave evidence that efficient sanitizers are inherently less powerful than their computationally unbounded counterparts.  They study the problem of construcing differentially private \textbf{one-shot sanitizers} that, given a database $D$, produce a summary from which approximate answers to \emph{every} query in $\cQ$ can be computed, while both the sanitizer and the summary run in time much less than the size of $\cQ$.  Dwork et al.~constructed a family of $2^{\tilde{O}(\sqrt{n})}$ queries for which there is no efficient (time $\poly(d,n)$) one-shot sanitizer (under certain cryptographic assumptions), even though there is an inefficient (time $\poly(2^d, n, |\cQ|)$) one-shot sanitizer even if $|\cQ|$ is nearly $2^n$.  For any family $\cQ$, constructing an efficient one-shot sanitizer is one way of constructing an efficient sanitizer that answers any polynomial number of queries from $\cQ$.  Thus, hardness results for one-shot sanitizers rule out a particular way of constructing efficient sanitizers.  However, ultimately a polynomial-time analyst will only be able to ask a polynomial number of queries, and hardness results for one-shot sanitzers still leave hope that there might be an efficient sanitizer that can answer many more arbitrary counting queries than the Laplace Mechanism.

Unfortunately, we show that this is not the case---there is no efficient, differentially private algorithm that takes a database $\db \in (\bits^d)^n$, and $\tilde{\Theta}(n^2)$ arbitrary, efficiently computable counting queries as input and outputs an approximate answer to each of the queries.  One way to summarize our results is that, unless we restrict the set $\cQ$ of allowable queries, or allow exponential running time, then the Laplace Mechanism is essentially the best possible algorithm for answering counting queries with differential privacy.

\subsection{Our Results and Techniques}

As discussed above, in this paper we give new hardness results for answering counting queries while satisfying differential privacy.  To make the statement of our results more concrete, we will assume that the counting queries are given to the sanitizer as input in the form of circuits that, on input an individual record $x \in \bits^d$, decide whether or not the record $x$ satisfies the property $q$.  We say the queries are efficiently computable if the corresponding circuits are of size $\poly(d,n)$.

\begin{theorem} \label{thm:main1}
Assuming the existence of one-way functions, there exists a $k = n^{2 + o(1)}$ such that there is no algorithm that, on input a database $D \in (\bits^d)^n$ and $k$ efficiently computable counting queries, runs in time $\poly(d,n)$ and returns an approximate answer to each query to within $\pm .49$, while satisfying differential privacy.
\end{theorem}

In particular, Theorem~\ref{thm:main1} applies to \textbf{interactive sanitizers}, which are sanitizers that receive (possibly adaptively chosen) queries one at a time.  Many positive results achieve this stronger notion of sanitization.  In particular, the Laplace mechanism is an efficient interactive sanitizer that answers $\tilde{\Omega}(n^2)$ queries and there exist interactive sanitizers that can answer nearly $2^{n}$ queries in time $\poly(2^d, n)$ per query interactively~\cite{RothRo10, HardtRo10, GuptaRoUl12}.

We also show that, the same theorem holds even for queries that are computable by unbounded-fan-in circuits of depth $6$ over the basis $\{\land, \lor, \neg\}$ (a subset of the well-studied class $AC^0$), albeit under a stronger (but still plausible) cryptographic assumption.

\begin{theorem} \label{thm:main2}
Under the assumptions described in Section~\ref{sec:decfams}, there exists a $k = n^{2 + o(1)}$ such that there is no algorithm that, on input a database $D \in (\bits^d)^n$ and $k$ efficiently computable depth-$6$ queries (circuits), runs in time $\poly(d,n)$ and returns an approximate answer to each query to within $\pm .49$, while satisfying differential privacy.
\end{theorem}

Theorem~\ref{thm:main2} should be contrasted with the results of Hardt, Rothblum, and Servedio~\cite{HardtRoSe12} as well as Thaler, Ullman, and Vadhan~\cite{ThalerUlVa12}, which give efficient sanitizers for answering $n^{\Omega(\sqrt{k})} \gg n^2$ monotone $k$-way conjunction queries, a much simpler class than polynomial-size depth-6 circuits.\footnote{A monotone $k$-way conjunction query on a database $D \in (\bits^d)^*$ is specified by a set of positions $S \subseteq [d]$, $|S| = k \leq d$, and asks ``What fraction of records in $D$ have a $1$ in every position in $S$?''.}  

We now describe our techniques.
\paragraph{The Connection with Traitor-Tracing}
We prove our results by building on the connection between differentially private sanitizers for counting queries and \emph{traitor-tracing schemes} discovered by Dwork et al.~\cite{DworkNaReRoVa09}.  Traitor-tracing schemes were introduced by Chor, Fiat, and Naor~\cite{ChorFiNa94} for the purpose of identifying pirates who violate copyright restrictions.  Roughly speaking, a (fully collusion-resilient) traitor-tracing scheme allows a sender to generate keys for $n$ users so that 1) the sender can broadcast encrypted messages that can be decrypted by any user, and 2) any \emph{efficient pirate decoder} capable of decrypting messages can be \emph{traced} to at least one of the users who contributed a key to it, even if an arbitrary coalition of the users combined their keys in an arbitrary efficient manner to construct the decoder.

Dwork et al.~show that the existence of traitor-tracing schemes implies hardness results for one-shot sanitizers.  Very informally, they argue as follows:  Suppose a coalition of users takes their keys and builds a database $\db \in (\bits^d)^n$ where each record contains one of their user keys.  The family $\cQ$ will contain a query $q_{c}$ for each possible ciphertext $c$.  The query $q_{c}$ asks ``What fraction of the records (user keys) in $D$ would decrypt the ciphertext $c$ to the message $1$?''  Every user can decrypt, so if the sender encrypts a message $m \in \bits$ as a ciphertext $c$, then every user will decrypt $c$ to $m$.  Thus the answer to the counting query, $q_{c}$, will be $m$.  

Suppose there were an efficient one-shot sanitizer for $\cQ$.  Then the coalition could use it to efficiently produce a summary of the database $D$ that enables one to efficiently compute an approximate answer to every query $q_{c}$, which would also allow one to efficiently decrypt the ciphertext.  Such a summary can be viewed as an efficient pirate decoder, and thus the tracing algorithm can use the summary to trace one of the users in the coalition.  However, if there is a way to identify one of the users in the database from the summary, then the summary is not differentially private.

In order to instantiate their result, they need a traitor-tracing scheme.  Since $\cQ$ contains a query for every ciphertext, the parameter to optimize is the length of the ciphertexts.  Using the fully collusion-resilient traitor-tracing scheme of Boneh, Sahai, and Waters~\cite{BonehSaWa06}, which has ciphertexts of length $\tilde{O}(\sqrt{n})$, they obtain a family of queries of size $2^{\tilde{O}(\sqrt{n})}$ for which there is no efficient one-shot sanitizer.  Dwork et al.~also discovered a partial converse---proving hardness of one-shot sanitization for a smaller family of queries requires constructing traitor-tracing schemes with shorter ciphertexts, which is a seemingly difficult open problem. 
\paragraph{Our Approach}
In our setting of sanitization (rather than one-shot sanitization, as studied by Dwork et al.~\cite{DworkNaReRoVa09}), we don't expect to answer every query in $\cQ$, only a much smaller set of queries requested by the analyst.  At first glance, this should make answering the queries much easier, and thus make it more difficult to demonstrate hardness.  However, the attacker does have the power to choose the queries  which he wants answered, and can choose queries that are most difficult to sanitize.  Our first observation is that in the traitor-tracing scenario, the tracing algorithms only query the pirate decoder on a polynomial number of ciphertexts, which are randomly chosen and depend on the particular keys that were instantiated for the scheme.  For many schemes, even $\tilde{O}(n^2)$ queries is sufficient.  Thus it would seem that the tracing algorithm could simply decide which queries it will make, give those queries as input to the sanitizer, and then use the answers to those queries to identify a user and violate differential privacy.

However, this intuition ignores an important issue.  Many traitor-tracing schemes (including~\cite{BonehSaWa06}) can only trace \emph{stateless} pirate decoders, which essentially commit to a response to each possible query (or a distribution over responses) once and for all.  For one-shot sanitizers, the private summary is necessarily stateless, and thus the result of Dwork et al.~can be instantiated with any scheme that allows tracing of stateless pirate decoders.  However, an arbitrary sanitizer might give answers that depend on the sequence of queries.  Thus, in order to prove our results, we will need a traitor-tracing scheme that can trace \emph{stateful} pirate decoders.

The problem of tracing stateful pirates is quite natural even without the implications for private data analysis.  Indeed, this problem has been studied in the literature, originally by Kiayias and Yung~\cite{KiayiasYu01}.  They considered pirates that can \emph{abort} and \emph{record history}.  However, their solution, and all others known, does not apply to our specific setting due to a certain ``watermarking assumption'' that doesn't apply when proving hardness-of-sanitization  (see discussion below).  To address this problem, we also refine the basic connection between traitor-tracing schemes and differential privacy by showing that, in many respects, fairly weak traitor-tracing schemes suffice to establish the hardness of preserving privacy.  In particular, although the pirate decoder obtained from a sanitizer may be stateful and record history, the accuracy requirement of the sanitizer means that the corresponding pirate decoder cannot abort, which will make it easier to construct a traitor-tracing scheme for these kinds of pirates.  Indeed, we construct such a scheme to establish Theorem~\ref{thm:main1}.  

The scheme also has weakened requirements in other respects, having nothing to do with the statefulness of the pirate or the tracing algorithm. These weakened requirements  allow us to reduce the complexity of the decryption, which means that the queries used by the attacker do not need to be arbitrary polynomial-size circuits, but instead can be circuits of constant depth, which allows us to establish Theorem~\ref{thm:main2}.  Another technical issue arises in that all $k$ queries must be given to the sanitizer at once, whereas tracing algorithms typically are allowed to query the pirate interactively.  However, we are able to show that the scheme we construct can be traced using one round of queries.
See Sections~\ref{sec:traitortracing} and~\ref{sec:attack} for a precise statement of the kind of traitor-tracing scheme that suffices and Section~\ref{sec:construction} for our construction.

Our construction is based on a well-known fully collusion resilient traitor-tracing scheme~\cite{ChorFiNa94}, but with a modified tracing algorithm.  The tracing algorithm uses \emph{fingerprinting codes}~\cite{BonehSh98, Tardos08}, which have been employed before in the context of traitor-tracing and content distribution, but our tracing algorithm is different from all those we are aware of.  The resulting scheme allows for tracing with a minimal number of non-adaptively chosen queries, achieves tracing without context-specific watermarking assumptions, simplifies the decryption circuit (at the expense of weakening the security parameters and functionality).  The restriction to non-aborting pirates may not be so natural in the setting of content distribution, which may explain why the scheme was not previously known.

\subsection{Related Work}

In addition to the hardness results for one-shot sanitizations~\cite{DworkNaReRoVa09}, which apply to arbitrary one-shot sanitizers, there are several hardness-of-sanitization results for restricted classes of sanitizers.  Dwork et al.~proved stronger hardness results for sanitizers whose output is a ``synthetic database''---roughly, a new database (of the same dimensions) that approximately preserves the answer to some set of queries.  Their results were extended by Ullman and Vadhan~\cite{UllmanVa11}, who showed that it is hard to generate a private synthetic database that is accurate for essentially any non-trivial family of queries, even $2$-literal conjunctions.

Gupta et al.~\cite{GuptaHaRoUl11} considered algorithms that can only access the database by making a polynomial number of ``statistical queries'' (essentially counting queries).  They showed that such algorithms cannot be a one-shot sanitizer (even ignoring privacy constraints) that approximately answers certain simple families of counting queries with high accuracy.

Finally, Dwork, Naor, and Vadhan~\cite{DworkNaVa12} gave information-theoretic lower bounds for \emph{stateless sanitizers}, which take $k$ queries as input, but whose answers to each query do not depend on the other $k-1$ input queries.  They showed that (even computationally unbounded) stateless sanitizers can answer at most $\tilde{O}(n^2)$ queries with non-trivial accuracy, while satisfying differential privacy.  The Laplace Mechanism is a stateless sanitizer that answers $\tilde{\Omega}(n^2)$ queries, and thus their result is tight in this respect.  Although their result is information theoretic, and considers a highly restricted type of sanitizer, their techniques are related to ours.  We elaborate on this connection in the appendix.

\section{Preliminaries} \label{sec:sans}

\paragraph{Differentially Private Algorithms}

Let a \emph{database} $\db \in \dbset$ be a collection of $n$ rows $(\rowseq) \in \univ$.  We say that two databases $\db, \db' \in \dbset$ are \emph{adjacent} if they differ only on a single row, and we denote this by $\db \sim \db'$.
\begin{definition}[Differential Privacy~\cite{DworkMcNiSm06}]\label{def:dp} Let $\san \from \dbset \to \cR$ be a randomized algorithm that takes a database as input (where $\rows$ and $\cols$ are varying parameters). $\san$ is \emph{$(\eps, \delta)$-differentially private} if for every two adjacent databases $\db \sim \db'$ and every subset $S \subseteq \cR$,
$$
\prob{\san(D) \in S} \leq e^{\eps} \prob{\san(D') \in S} + \delta.
$$
If $\san$ is $(\eps, \delta)$-differentially private for some functions $\eps = \eps(\rows) = O(1)$, $\delta = \delta(\rows) = o(1/\rows)$, we will drop the parameters $\eps$ and $\delta$ and say that $\san$ is \emph{differentially private}.
\end{definition}
The choice of $\eps = O(1), \delta = o(1/\rows)$ is a fairly weak setting of the privacy parameters, and most known constructions of differentially private mechanisms for various tasks give quantitatively stronger privacy guarantees.  These parameters are essentially the weakest possible, as $(\eps, \delta)$-differentially privacy is not a satisfactory privacy guarantee for $\eps = \omega(1)$ or $\delta = \Omega(1/\rows)$.  That our lower bounds apply to the parameters specified in Definition~\ref{def:dp} makes our results stronger.

\paragraph{Sanitizers for Counting Queries}

Since an algorithm that always outputs $\bot$ satisfies Definition~\ref{def:dp}, we also need to specify what it means for the sanitizer to be useful.  In this paper we study sanitizers that give accurate answers to \emph{counting queries}.  A counting query on $\bitslen{\cols}$ is defined by a predicate $\query\from \bitslen{\cols} \to \bits$.  Abusing notation, we define the evaluation of the query $\query$ on a database $\db = (\rowseq) \in \dbset$ to be
$$
\query(\db) = \frac{1}{\rows} \sum_{i=1}^{\rows} \query(\rowi{i})
$$
We will use $\querysetd{\cols}$ to denote a set of counting queries on $\bitslen{\cols}$ and $\queryset = \bigcup_{\cols \in \N} \querysetd{\cols}$.

We are interested in \emph{sanitizers} that take a sequence of queries from some set $\queryset$ as input.  Formally a sanitizer is a function $\san \from \dbset \times (\querysetd{\cols})^\queries \to \mathbb{R}^\queries$ (where, again, $n, d$, and $\queries$ are varying parameters).  Notice that we assume that $\san$ outputs $k$ real-valued answers.  Think of the $j$-th component of the output of $\san$ as an answer to the $j$-th query.  For the results in this paper, this assumption will be without loss of generality.\footnote{In certain settings, $\san(\db, \queryseq)$ is allowed to output a ``summary'' $z \in \cR$ for some range $\cR$.  In this case, we would also require that there exists an ``evaluator'' $\cE \from \cR \times \queryset \to \R$ that takes a summary and a query and returns an answer $\cE(z,\query) = \answer$ that approximates $\query(\db)$.  The extra generality is used to allow $\san$ to run in less time than the number of queries it is answering (e.g.~releasing a fixed family of queries), but this is not relevant for our range of parameters where $\queries = \tilde{O}(\rows^2)$.  Indeed, a generic sanitizer, $\san$ that outputs a summary in $\cR$ can be converted to a generic sanitizer with output in $\R^\queries$ simply by running $\san(\db, \queryseq)$ to obtain $z \in \cR$ and then computing $\answerj{1} = \cE(z, \queryj{1}), \dots, \answerj{\queries} = \cE(z, \queryj{\queries})$.  This conversion increases the running time by a factor of $k$, which is the minimum time required to read the input queries.}  Definition~\ref{def:dp} extends naturally to sanitizers by requiring that for every $\queryseq \in \queryset$, the sanitizer $\san_{\queryseq}(\cdot) = \san(\cdot, \queryseq)$ is $(\eps, \delta)$-differentially private as a function of the input database.

Now we formally define what it means for a sanitizer to give accurate answers.
\begin{definition}[Accuracy]\label{def:acc}
Let $\db$ be a database and $\queryseq$ be a set of counting queries.  A sequence of answers $\answerseq$ is \emph{$\alpha$-accurate for $\queryseq$ on $\db$} if
$$
\forall j \in [k], | \queryj{j}(\db) - \answerj{j} | \leq \acc.
$$
 Let $\queryset$ be a set of counting queries, $k \in \N$ and $\alpha, \beta \in [0,1]$ be parameters.  A sanitizer $\san$ is \emph{$(\acc, \accfail, \queryset, \queries)$-accurate} if for every database $\db \in (\bits^\cols)^{\rows}$ and every sequence of queries $\queryseq \in \querysetd{\cols}$
$$
\Prob{\textrm{$\san$'s coins}}{\san(\db, \queryseq) \textrm{ is $\acc$-accurate for $\db$ and $\queryseq$}} \geq 1- \accfail.
$$
If $\san$ is $(\acc, \accfail, \queryset, \queries)$-accurate for any (constant) $\acc < 1/2$ and $\accfail = \accfail(\rows) = o(1/n^2)$, we will drop $\acc$ and $\accfail$ and say that $\san$ is $(\queryset, \queries)$-accurate.
\end{definition}
The choice of $\acc < 1/2, \accfail = o(1/\rows^2)$ is also significantly weaker than what can be achieved by known constructions of sanitizers.  These parameters are also essentially the weakest parameters possible, as a mechanism that answers $1/2$ to every query achieves $\acc = 1/2$, $\accfail = 0$ for any number of arbitrary queries.  Also, if there is a mechanism that achieves $(\alpha, \accfail, \queryset, \queries)$-accuracy for $\accfail < 1/2$, then there is another mechanism that achieves $(\alpha, o(1/\rows^2), \queryset, \queries)$-accuracy with only an $O(\log n)$ loss in the privacy parameters and the efficiency of the mechanism.\footnote{Given a sanitizer $\san$ that answers every query accurately with probability $1/2 + \Omega(1)$, one can obtain a mechanism $\san'$ that answers every query accurately with probability $1-\beta$.  $\san'$ will run $\san$ independently $r = O(\log(1/\beta))$ times and answers each query with the median of the $r$ answers for that query.} That our lower bound applies to the parameters specified in Definition~\ref{def:acc} makes our results stronger. 

\paragraph{Efficiency of Sanitizers}
Simply, a sanitizer is efficient if it runs in time polynomial in the length of its input.  To make the statement more precise, we need to specify how the queries are given to the sanitizer as input.  

Notice that to specify an arbitrary counting query $\query \from \bitslen{\cols} \to \bits$ requires $2^\cols$ bits.  In this case, a sanitizer whose running time is polynomial in the time required to specify the query is not especially efficient.  Thus, we restrict attention to queries that are efficiently computable, and have a succinct representation.  In this work, we will fix the representation to be Boolean circuits over the basis $\set{\land, \lor, \neg}$ with possibly unbounded-fan-in.  In this representation, any query can be evaluated in time $|\query|$, where $|\cdot|$ denotes the size of the circuit computing $\query$.  We also want to consider the case where the queries are computable by circuits of low depth.  For a constant $h \in \N$, we use $\querysetd{\cols}_{\mathsf{depth-}h}$ to denote the set of all counting queries on $\bitslen{\cols}$ specified by circuits of depth $h$.  Finally, we use $\querysetd{\cols}_{\mathsf{all}}$ to denote the set of all counting queries on $\bitslen{\cols}$.

\begin{definition} [Efficient Sanitizers]
A sanitizer $\san$ is \emph{efficient} if, on input a database $\db \in (\bits^\cols)^\rows$ and $\queries$ queries $\queryj{1}, \dots, \queryj{\queries} \in \querysetd{\cols}_{\mathsf{all}}$, $\san$ runs in time $\poly(\cols, \rows, \queries, |\queryj{1}| + \dots + |\queryj{\queries}|)$.  For every $h \in \N$, a sanitizer $\san$ is \emph{efficient for depth-$h$} queries if, on input a database $\db \in (\bits^\cols)^\rows$ and $\queries$ queries $\queryj{1}, \dots, \queryj{\queries} \in \querysetd{\cols}_{\mathsf{depth}-h}$, $\san$ runs in time $\poly(\cols, \rows, \queries, |\queryj{1}| + \dots + |\queryj{\queries}|)$.
\end{definition}

For comparison with our results, we will recall the properties of some known mechanisms, stated in our terminology and for our choice of parameters:
\begin{theorem}[Laplace Mechanism~\cite{DinurNi03, BlumDwMcNi05, DworkMcNiSm06}]
There exists a sanitizer $\san_{\mathsf{Lap}}$ that is 1) differentially private, 2) efficient, and 3) $(\querysetd{\cols}_{\mathsf{all}}, \tilde{\Omega}(n^2))$-accurate.
\end{theorem}
\begin{theorem}[``Advanced Query Release Mechanisms''~\cite{BlumLiRo08, DworkNaReRoVa09, DworkRoVa10, HardtRo10, GuptaRoUl12, HardtLiMc12}]
There exists a sanitizer $\san_{\mathsf{Adv}}$ that is 1) differentially private and 2) $(\querysetd{\cols}_{\mathsf{all}}, 2^{\tilde{\Omega}(n/ \sqrt{d})})$-accurate.  For queries $\queryj{1}, \dots, \queryj{\queries} \in \querysetd{\cols}_{\mathsf{all}}$, $\san_{\mathsf{Adv}}$ runs in time $\poly(2^\cols, \rows, \queries, |\queryj{1}| + \dots + |\queryj{\queries}|)$.
\end{theorem}
As we mentioned above, these mechanisms can achieve stronger quantitative privacy and accuracy guarantees (in terms of $\eps, \delta$ for privacy and $\acc, \accfail$ for accuracy) with only a small degradation in the number of queries.  Also, notice that both of these mechanisms provide accuracy guarantees that are independent of the complexity of the queries (although the running time of the mechanism will depend on the complexity of the queries).  Our hardness results will apply to sanitizers that only provide accuracy for queries of size $\poly(\cols, \rows)$.

\section{Traitor-Tracing Schemes} \label{sec:tt}
In this section we give define a traitor-tracing scheme.
Throughout, we will use $\mathsf{A_{TT}}$ to denote algorithms associated with traitor-tracing schemes.

\subsection{Traitor-Tracing Schemes} \label{sec:traitortracing}
We now give a definition of a traitor-tracing scheme, heavily tailored to the task of proving hardness results for generic sanitizers. We will sacrifice some consistency with the standard definitions.  See below for further discussion of the ways in which our definition departs from the standard definition of traitor-tracing.  In some cases, the non-standard aspects of the definition will be necessary to establish our results, and in others it will be for convenience.  Despite these differences, we will henceforth refer to schemes satisfying our definition simply as \emph{traitor-tracing schemes}.

 Intuitively, a traitor-tracing scheme is a form of broadcast encryption, in which a sender can broadcast an encrypted message that can be decrypted by each of a large set of users. The standard notion of security for such a scheme would require that an adversary that doesn't have any of the keys cannot decrypt the message.  A traitor-tracing scheme has the additional property that given any decoder capable of decrypting the message (which must in a very loose sense ``know'' at least one of the keys), there is a procedure for determining which user's key is being used.  Moreover, we want the scheme to be ``collusion resilient,'' in that even if a coalition of users gets together and combines their keys in some way to produce a decoder, there is still a procedure that identifies at least one member of the coalition.

We now describe the syntax of a traitor-tracing scheme more formally.  For functions $\users, \tracequeries \from \N \to \N$, an $(\users, \tracequeries)$-traitor-tracing scheme is a tuple of algorithms $(\ttgen,  \ttenc,  \ttdec, \trace)$.  The parameter $n$ specifies the number of users of the scheme and the parameter $\tracequeries$ will specify the number of queries that the tracing algorithm makes to the pirate decoder.  We allow all the algorithms to be randomized except for $\ttdec$.\footnote{It would not substantially affect our results if $\ttdec$ were randomized, but we will assume that $\ttdec$ is deterministic for ease of presentation.}
\begin{itemize}
\item The algorithm $\ttgen$ takes a security parameter, $\param$, and returns a sequence of $\users = \users(\param)$ user keys $\userkeyvec = (\userkeys) \in \bits^{\param}$.  Formally, $\userkeyvec = (\userkeys) \getsr \ttgen(1^{\param})$.
\item The algorithm $ \ttenc$ takes a sequence of $\users$ user keys $\userkeyvec$ and a message bit $b \in \bits$, and generates a ciphertext $\ctext \in \ctexts = \ctexts^{(\param)}$.  Formally, $\ctext \getsr  \ttenc(\userkeyvec, b)$.
\item The algorithm $\ttdec$ takes any single user key $\userkey$ and a ciphertext $\ctext \in \ctexts$, runs in time $\poly(\param, \users(\param))$ and deterministically returns a message bit $\decbit \in \bits$. Formally $\decbit = \ttdec(\userkey, \ctext)$.
\item The algorithm $\trace$ takes as input a set of user keys $\userkeyvec \in (\bits^\param)^{\users(\param)}$ and an oracle $\pirate \from (\ctexts^{(\param)})^{\tracequeries(\param)} \to \bits^{\tracequeries(\param)}$, makes one $\tracequeries$-tuple of queries, $(\ctextseq) \in \ctexts^{(\param)}$ to its oracle ($\tracequeries = \tracequeries(\param)$), and returns the name of a user $i \in [\users(\param)]$.  Formally, $i \getsr \tracep(\userkeyvec)$.
\end{itemize}

Intuitively, think of the oracle $\pirate$ as being given some subset of keys $\userkeyvec_{S} = (\userkeyi{i})_{i \in S}$ for a non-empty set $S \subseteq [\users]$, and $\trace$ is attempting to identify a user $i \in S$.  Clearly, if $\pirate$ ignores its input and always returns $0$, $\trace$ cannot have any hope of success, so we must assume that $\pirate$ is capable of decrypting ciphertexts.

\begin{definition}[Available Pirate Decoder] \label{def:availablepirate}
Let the tuple of algorithms \linebreak $\ttscheme = (\ttgen, \ttenc, \ttdec, \trace)$ be an $(\users, \tracequeries)$-traitor-tracing scheme.  Let $\pirate$ be a (possibly randomized) algorithm.  We say that $\pirate$ is a \emph{$\tracequeries$-available pirate decoder} if for every $\param \in \N$, every set of user keys $\userkeyvec = (\userkeys) \in \bits^{\param}$, every $S \subseteq [\users]$ such that $|S| \geq \users-1$, and every $\ctextseq \in \ctexts^{(\param)}$,
\begin{equation*}
\Prob{}{(\decbitseq) \getsr \pirate(\userkeyvec_{S}, \ctextseq) \atop \exists j \in [\tracequeries], b \in \bits \left( \left( \forall i \in S, \ttdec(\userkeyi{i}, \ctextj{j}) = b\right) \land \left( \decbitj{j} \neq b \right) \right)} \leq o\left(\frac{1}{\users(\param)^2}\right).
\end{equation*}
In other words, if every user key $\userkeyi{i}$ (for $i \in S$) decrypts $\ctext$ to $1$ (resp.~$0$), then $\pirate(\userkeyvec_{S}, \cdot)$ decrypts $\ctext$ to $1$ (resp.~$0$), with high probability.
\end{definition}

We can now define a secure, $(\users, \tracequeries)$-traitor-tracing scheme:
\begin{definition}[Traitor-Tracing Scheme] \label{def:ttscheme}
Let the tuple of algorithms \linebreak $\ttscheme = (\ttgen, \ttenc, \ttdec, \trace)$ be an $(\users, \tracequeries)$-traitor-tracing scheme.  Let $\tracequeries \from \N \to \N$ be a function.  We say that $\ttscheme$ is a \emph{secure $(\users, \tracequeries)$-traitor-tracing scheme} if for every $S \subseteq [\users(\param)]$ such that $|S| \geq \users(\param) - 1$, for every (possibly randomized) algorithm $\pirate$ that 1) runs in time $\poly(\param, \users(\param), \tracequeries(\param))$ and 2) is a $\tracequeries$-available pirate decoder, we have
\begin{equation*}
\Prob{\userkeyvec \getsr \ttgen(1^\param) \atop \textrm{$\pirate$'s, $\trace$'s coins}}{\mathsf{Trace}_{\mathsf{TT}}^{\pirate(\userkeyvec_{S}, \cdot)}(\userkeyvec) \not\in S} = o\left(\frac{1}{\users(\param)}\right)
\end{equation*}
\end{definition}

\paragraph{Remarks About Our Definition of Traitor-Tracing}
The traitor-tracing schemes we consider are somewhat different than those previously studied in the literature.
\begin{itemize}
\item We do not require the encryption or tracing algorithms to use public keys.  In the typical application of traitor-tracing schemes to content distribution, these would be desirable features, however they are not necessary for proving hardness of sanitization.

\item We only require that the tracing algorithm succeeds with probability $1 - o(1/n)$.  Typically one would require that the tracing algorithm succeeds with probability $1 - n^{-\omega(1)}$.

\item We do not give the pirate decoder access to an encryption oracle.  In other words, we do not require CPA security.  Most traitor-tracing schemes in the literature are public-key, making this distinction irrelevant.  Here, we only need an encryption scheme that is secure for an \emph{a priori} bounded number of messages.

\item We allow the pirate decoder to be \emph{stateful}, but in an unusual way.  We require (roughly) that if any of the queries are ciphertexts generated by $\enc(\userkeyvec, b)$, then the pirate decoder answers $b$ to those queries, regardless of the other queries issued.  In many models, the pirate is allowed to abort, and answer $\bot$ if it detects that it is being traced.  However, we do allow our pirate to correlate its answers to different queries, subject to this accuracy constraint.  We also allow the pirate to see all the queries made by the tracer at once, which is more power than is typically given to the pirate.
\end{itemize}

Roughly, the first three modifications will allow us to find a candidate scheme with very simple decryption and the fourth modification will allow us to trace stateful pirates even in the setting of bit-encryption.
 
\subsection{Decryption Function Families} \label{sec:decryptionfamilies}

For Theorem~\ref{thm:main2}, we are interested in traitor-tracing schemes where $\ttdec$ is a ``simple'' function of the user key (for every ciphertext $\ctext \in \ctexts$).
\begin{definition}[Decryption Function Family]
Let $(\ttgen, \ttenc, \ttdec)$ be a traitor-tracing scheme where $\ttgen$ produces keys in $\bits^\param$ and $\ttenc$ produce ciphertexts in $\ctexts = \ctexts^{(\param)}$.  For every $\ctext \in \ctexts$, we define the \emph{$\ctext$-decryption function} $\decfn_{\ctext} \from \bits^\param \to \bits$ to be $\decfn_{\ctext}(\userkey) = \ttdec(\userkey, \ctext)$.  We define the \emph{decryption function family} $\decfns_{\ttdec}^{(\param)} = \set{\decfn_{\ctext}}_{\ctext \in \ctexts^{(\param)}}$.
\end{definition}

In what follows, we will say that $\ttscheme$ is an traitor-tracing scheme with decryption function family $\decfns_{\ttdec}^{(\param)}$.
 
 \section{Attacking Efficient Sanitizers} \label{sec:attack}
In this section we will prove our main result, showing that the existence of traitor-tracing schemes (as in Definition~\ref{def:ttscheme}) implies that efficient sanitizers cannot answer too many counting queries while satisfying differential privacy.
 
 \begin{theorem} [Attacking Efficient Sanitizers] \label{thm:attack}
 Assume that there exists an $(\users(\param), \tracequeries(\param))$-secure traitor-tracing scheme $\ttscheme = (\ttgen, \ttenc, \ttdec, \trace)$ with decryption function family $\queryset^{(\param)} = \decfns_{\ttdec}^{(\param)}$.  Then there does not exist any sanitizer $\san \from (\bits^\cols)^\rows \times (\queryset^{(\cols)})^{\tracequeries(\cols)} \to \R^{\tracequeries(\cols)}$ that is simultaneously 1) differentially private, 2) efficient, and 3) $(\queryset, \tracequeries(\cols))$-accurate for $\queryset = \cup_{\cols \in \N} \queryset^{(\cols)}$.
 \end{theorem}
 
 In the typical setting of parameters, $\users(\param) = \poly(\param)$, $\tracequeries(\param) = \tilde{\Theta}(\users^2)$, and decryption can be implemented by circuits of size $\poly(\users) = \poly(\param)$.  Then Theroem~\ref{thm:attack} will state that there is no sanitizer $\san$ that takes a database $\db \in (\bits^\cols)^{\poly(\cols)}$, runs in $\poly(\cols)$ time, and accurately answers $\tilde{\Theta}(\rows^2)$ queries implemented by circuits of size $\poly(\cols)$, while satisfying differential privacy.  
 
The main difference between Theorem~\ref{thm:attack} and the result of Dwork et al.~\cite{DworkNaReRoVa09} is that we only assume the existence of a sanitizer for $\tracequeries(\cols)$ queries from $\queryset^{(\cols)} = \queryset_{\ttdec}^{(\cols)}$, whereas Dwork et al.~assume the existence of a one-shot sanitizer that answers every query in $\decfns^{(\cols)}$.  To offset the weaker assumption on the sanitizer, we assume that the traitor-tracing scheme is secure against certain stateful pirate decoders (as in Definition~\ref{def:availablepirate}) whereas Dwork et al.~only need to trace stateless pirates.  Theorem~\ref{thm:attack} also explicitly allows the traitor-tracing scheme to have the relaxed functionality and security properties discussed at the end of Section~\ref{sec:tt}, although it is implicit in Dwork et al.~that the relaxed properties are sufficient to prove hardness results.
\jnote{When I say ``is implicit in'', I mean that we use the same arguments but note that weaker assumptions suffice.  Does that match with the usual use of ``is implicit in''?}
 
We now sketch the proof:  Every function $q_{\ctext} \in \decfns^{(\cols)}$ is viewed as a query $q_{\ctext}(x)$ on a database row $x \in \bits^\cols$.  Assume there is an efficient sanitizer is that is $(\decfns^{(\cols)}, \tracequeries(\cols))$-accurate for these queries.  The fact that $\san$ is accurate for these queries will imply that (after small modifications) $\san$ is a $\tracequeries$-available pirate decoder (Definition~\ref{def:availablepirate}).  Here is where we differ from Dwork et al., who assume that $\san$ accurately answers \emph{all} queries in $\querysetd{\cols}$, in which case $\san$ can be viewed as a stateless pirate decoder (but must solve a harder sanitization problem).

We complete the proof as in Dwork et al.  Consider two experiments:  In the first, we construct an $\users$-row database $\db$ by running $\ttgen(1^{\cols})$ to obtain $\users$ user keys, and putting one in each row of $\db$.  Then we run $\trace$ on $\san(\db, \cdot)$ and obtain a user $i$.  Since $\san$ is useful, and $\ttscheme$ is secure, we will have that $i \in [\users]$ with probability close to $1$, and thus there is an $i^* \in [\users]$ such that $i = i^*$ with probability $\gtrsim 1/\users$.  

In the second experiment, we construct a database $\db'$ exactly as in the first, however we exclude the key $\userkeyi{i^*}$.  Since $\db$ and $\db'$ differ in only one row, differential privacy requires that $\trace$, run with oracle $\san(\db', \cdot)$, still outputs $i^*$ with probability $\Omega(1/\users)$.  However, in this experiment, $i^*$, $\userkeyi{i^*}$ is no longer given to the pirate decoder, and thus security of $\ttscheme$ says that $\trace$, run with this oracle, must output $i^*$ with probability $o(1/\users)$.  Thus, we will obtain a contradiction.
 
 \begin{proof}
Let $\ttscheme = (\ttgen, \ttenc, \ttdec, \trace)$ be the assumed traitor-tracing scheme, and assume there exists an efficient, differentially private, $(\queryset^{(\cols)}, \tracequeries(\cols))$-accurate sanitizer $\san$. 
We define the pirate decoder $\pirate_{\san}$ as follows:
\begin{algorithm}
\caption{The pirate decoder $\pirate_{\san}$}
\begin{algorithmic}
\STATE{\textbf{Input:} A set of user keys $(\userkeyvec_{S}) \in \bits^{\cols}$ and a set of ciphertexts $\ctextseq$ ($\tracequeries = \tracequeries(\cols)$).}
\STATE{Construct circuits specifying the queries $q_{\ctextj{1}}, \dots, q_{\ctextj{\tracequeries}} \in \decfns_{\ttdec, \cols}$.}
\STATE{Construct a database $\db = (\userkeyi{i})_{i \in S} \in (\bits^{\cols})^{|S|}$.}
\STATE{Let $\ttanswerseq \getsr \san(\db, q_{\ctextj{1}}, \dots, q_{\ctextj{\tracequeries}})$.}
\STATE{Round the answers $\ttanswerseq \in [0,1]$ to obtain $\decbitj{1}, \dots, \decbitj{\tracequeries} \in \bits$ (i.e. $\decbitj{j} = \lceil a_j \rfloor$)}
\STATE{\textbf{Output:}  $\decbitj{1}, \dots, \decbitj{\tracequeries}$.}
\end{algorithmic}
\end{algorithm}

Since $\san$ is efficient, its running time is $\poly(\cols, \users(\cols), \tracequeries(\cols), |q_{\ctextj{1}}| + \ldots + |q_{\ctextj{\tracequeries}}|)$, which is  $\poly(\cols, \users(\cols), \tracequeries(\cols))$.  Recall that the size of the circuits (queries) $q_{\ctext} \in \decfns_{\ttdec}^{(\cols)}$ is $\poly(\cols, \rows)$.
 In this case $\pirate_{\san}$ runs in time $\poly(\cols, \users(\cols), \tracequeries(\cols))$ as well, since constructing the queries can be done in time polynomial in their size, and the final rounding step can be done in time $\poly(\tracequeries(\cols))$.

Next, we claim that if $\san$ is accurate for $\queryset^{(\cols)}$, then $\pirate_{\san}$ is a useful pirate decoder.
\begin{claim} \label{clm:acctouseful}
If $\san$ is $(\queryset_{\ttdec}, \tracequeries)$-accurate, then $\pirate_{\san}$ is a $\tracequeries$-useful pirate decoder.
\end{claim}
\begin{proof} [Proof of Claim~\ref{clm:acctouseful}]
Let $\userkeyvec \in \bits^\cols$ be a set of user keys for $\ttscheme$ and let $S \subseteq [\users]$ be a subset of the users such that $|S| \geq \users-1$.  Suppose $\ctext \in \ctexts^{(\cols)}$ and $b \in \bits$ are such that for every $i \in S$, $\ttdec(\userkeyi{i}, \ctext) = b$.  Then we have that, for $\db$ as in $\pirate_{\san}$,
$$
\query_{\ctext}(D) = \frac{1}{|S|} \sum_{i \in S} \query_{\ctext}(\userkeyi{i}) = \frac{1}{|S|} \sum_{i \in S} \ttdec(\userkeyi{i}, \ctext) = b
$$
Let $\ctextseq$ be a set of ciphertexts, $q_{\ctextj{1}}, \dots, q_{\ctextj{\tracequeries}}$ and $\ttanswerseq$ be as in $\pirate_{\san}$.  The accuracy of $\san$ (with constant error $\alpha < 1/2$) guarantees that
\begin{equation*}
\prob{\exists j \in [\tracequeries], \left| \answerj{j} - f_{\ctextj{j}}(\db) \right| \geq 1/2 } = o(1/|S|^2)
\end{equation*}
Since $|S| \geq \rows - 1$, $o(1/|S|^2) = o(1/\rows^2)$.  Assuming $\ttanswerseq$ is accurate up to error $\alpha < 1/2$ for $q_{\ctextj{1}}, \dots, q_{\ctextj{\tracequeries}}$,
$a_j$ will be rounded to exactly $q_{\ctextj{j}}$ whenever $q_{\ctextj{j}}(D) \in \bits$.
That is,
\begin{equation*}
\prob{\exists j \in [\tracequeries], b \in \bits \atop \left( \forall i \in S, \ttdec(\userkeyi{i}, \ctextj{j}) = b \right) \land \left(\widehat{b}_j \neq b \right)} = o\left(\frac{1}{\users(\param)^2}\right)
\end{equation*}
Thus, $\pirate_{\san}$ is $\tracequeries$-useful.  This completes the proof of the claim.
\end{proof}

Since $\pirate_{\san}$ is a $\tracequeries$-useful pirate decoder, and $\ttscheme$ is a $(\users, \tracequeries)$-secure traitor-tracing scheme, running $\trace$ on $\pirate_{\san}$ will always return some user $i \in [\users]$.  Thus there must be some user $i^*$ that $\trace$ returns with probability $\gtrsim 1/n$.
Specifically, for every $\param \in \N$, there exists $i^*(\param) \in [\users(\param)]$ such that,
 \begin{equation} \label{eq:detect2}
 \Prob{\userkeyvec \getsr \ttgen(1^{\param}) \atop \textrm{$\pirate_{\san}$'s, $\trace$'s coins}}{\mathsf{Trace}_{\mathsf{TT}}^{\pirate_{\san}(\userkeyvec, \cdot)}(\userkeyvec) = i^*(\param)} \geq \frac{1}{\users(\param)} - o\left( \frac{1}{\users(\param)} \right).
 \end{equation}
 
Let $S(\param) = [\users(\param)] \setminus \{ i^*(\param) \}$
Now we claim that if $\san$ is differentially private, then $\trace$ will output $i^*(\param)$ with significant probability, even $\pirate_{\san}$ is not given the key of user $i^*(\param)$.
\begin{claim} \label{clm:privacytodetection}
If $\san$ is differentially private (for $\eps = O(1)$, $\delta = o(1/\rows)$), then
\begin{equation*}
 \Prob{\userkeyvec \getsr \ttgen(1^{\param}) \atop \textrm{$\pirate_{\san}$'s, $\trace$'s coins}}{\mathsf{Trace}_{\mathsf{TT}}^{\pirate_{\san}(\userkeyvec, \cdot)}(\userkeyvec) = i^*(\param)} \geq \Omega\left( \frac{1}{\users(\param)} \right).
\end{equation*}
\end{claim}
\begin{proof} [Proof of Claim~\ref{clm:privacytodetection}]
Fix any $\param$ and let $\tracequeries = \tracequeries(\param)$ and $i^* = i^*(\param)$, $S = S(\param)$ as above.  Let $\db = \userkeyvec$ and $\db_{-i^*} = \userkeyvec_{S}$.  Take $T$ to be the set of responses $\decbitseq$ such that $\trace(\userkeyvec)$, after querying its oracle on ciphertexts $\ctextseq$ and receiving responses $\decbitseq$, outputs $i^*$ ($T$ depends on the coins of $\ttgen$ and $\trace$).
By differential privacy, we have that
$$
\prob{\san(\db, \query_{\ctextj{1}}, \dots, \query_{\ctextj{\tracequeries}}) \in T} \leq e^{O(1)} \cdot \prob{\san(\db_{-i^*}, \query_{\ctextj{1}}, \dots, \query_{\ctextj{\tracequeries}}) \in T} + o\left(\frac{1}{\users} \right).
$$
Note that the queries constructed by $\pirate_{\san}$ depends only on $\ctextseq$, not on $\userkeyvec_{S}$.  Also note that the final rounding step does not depend on the input at all.  Thus, for every $T \subseteq \bits^{\tracequeries}$
\begin{equation} \label{eq:privacy1}
\prob{\pirate_{\san}(\userkeyvec, \ctextseq) \in T} \leq e^{O(1)} \cdot \prob{\pirate_{\san}(\userkeyvec_{S}, \ctextseq) \in T} + o\left( \frac{1}{\users}\right).
\end{equation}
The claim follows by combining with~\eqref{eq:detect2}.
\end{proof}
To complete the proof, notice that the probability in Claim~\ref{clm:privacytodetection} is exactly the probability that $\trace$ outputs the user $i^*$, when given the oracle $\pirate_{\san}(\userkeyvec_{S})$, for $S = [\users] \setminus \set{i^*}$.  However, the fact that $\pirate_{\san}$ is efficient, and $\ttscheme$ is a secure traitor-tracing scheme implies that this probability is $o(1/\users)$.  Thus we have obtained a contradiction.  This completes the proof of the Theorem.
 \end{proof}
 
  \section{Constructions of Traitor-Tracing Schemes} \label{sec:construction}
 In this section we show how to construct traitor-tracing schemes that satisfy Definition~\ref{def:ttscheme}, and thus can be used to instantiate Theorem~\ref{thm:attack}.  First we will informally describe a simple construction that requires the tracing algorithm to make a sub-optimal number of queries, but will hopefully give the reader more intuition about the construction and how it differs from previous constructions of traitor-tracing schemes.  Then we will give precise definitions of the encryption schemes (Section~\ref{sec:weakencryption}) and fingerprinting codes (Section~\ref{sec:fpcodes}) required for our construction.  Then we will present the final construction more formally (Section~\ref{sec:thescheme}) and prove its security.  Finally, we will use the weakened security requirements of the encryption scheme to show that our traitor-tracing scheme can be instantiated so that decryption is computable by constant-depth circuits (Section~\ref{sec:decfams}).
  
 \subsection{A Simple Construction} \label{sec:sketch}
Our construction is a variant of the most basic tracing traitor-tracing scheme~\cite{ChorFiNa94}.  Start with any encryption scheme \linebreak $(\gen, \enc, \dec)$.  Generate an independent key $\userkeyi{i} \getsr \gen$ for each user (we will ignore the security parameter in the informal description).  To encrypt a bit $b \in \bits$, we encrypt it under each user's key independently and concatenate the ciphertexts.  That is 
$$\ttenc(\userkeyvec, b) = (\enc(\userkeyi{1}, b),  \dots, \enc(\userkeyi{\users}, b)).$$  
Clearly each user can decrypt the ciphertext by applying $\dec$, as long as she knows which part of the ciphertext to decrypt.
\jnote{I wrote some text explaining the ``linear scan'' tracing alg and why it doesn't work in our setting.  I think it makes the section long and doesn't add much, but maybe I should include it?}

Now we describe how an available pirate decoder for this scheme can be traced.  As with all traitor-tracing schemes, we will form ciphertexts that different users would decrypt differently, assuming they decrypt as intended using the algorithm $\ttdec(\userkeyi{i}, \cdot)$.  We can do so with the following algorithm:
$$\malenc(\userkeyvec, i) = (\enc(\userkeyi{1}, 1),  \dots, \enc(\userkeyi{i}, 1), \enc(\userkeyi{i+1}, 0) , \dots,  \enc(\userkeyi{\users}, 0)$$
for $i = 0,1,\dots,\users$.  The algorithm forms a ciphertext that users $1, \dots, i$ will decrypt to $0$ and users $i+1,\dots, \users$ will decrypt to $1$.  

The tracing algorithm generates a random sequence $i_1, \dots, i_{\tracequeries} \in \set{0,1,\dots,\users}$, for $\tracequeries = (n+1)s$, such that each element of $\set{0,1,\dots,\users}$ appears exactly $s$ times, where $s$ is a parameter to be chosen later.  Then, for every $j$ it generates a ciphertext $\ctext_{j} \getsr \malenc(\userkeyvec, i_j)$.  Next, it queries $\pirate_{\userkeyvec_{S}}(\ctextj{1}, \dots, \ctextj{\tracequeries})$.  Given the output of the pirate, the tracing algorithm computes
$$P_i = \frac{1}{s} \sum_{j : i_j = i} \pirate(\userkeyvec, \ctext_{1}, \dots, \ctext_{\tracequeries})_j$$
for $i = 0,1,\dots, \users$.  Finally, the tracing algorithm outputs any $i^*$ such that $P_{i^*} - P_{i^*-1} \geq 1/\users$.

The tracing algorithm generates a random sequence of indices $i_1, \dots, i_{\tracequeries} \in \set{0,1,\dots,\users}$, for $\tracequeries = (n+1)s$, such that each element of $\set{0,1,\dots,\users}$ appears exactly $s$ times, where $s$ is a parameter to be chosen later.  Then, for every $j$ it generates a ciphertext $\ctext_{j} \getsr \malenc(\userkeyvec, i_j)$.  Next, it queries $\pirate_{\userkeyvec_{S}}(\ctextj{1}, \dots, \ctextj{\tracequeries})$.  Given the output of the pirate, the tracing algorithm computes
$P_i = \frac{1}{s} \sum_{j : i_j = i} \pirate(\userkeyvec, \ctext_{1}, \dots, \ctext_{\tracequeries})_j$
for $i = 0,1,\dots, \users$.  Finally, the tracing algorithm outputs any $i^*$ such that $P_{i^*} - P_{i^*-1} \geq 1/\users$.

Now we explain why this algorithm successfully traces efficient available pirate decoders.  Notice that if we choose $\ctext$ according to $\malenc(\userkeyvec, 0)$, then every user decrypts $\ctext$ to $0$, so $P_0 = 0$.  Similarly, $P_{\users} = 1$.  Thus, there exists $i^*$ such that $P_{i^*} - P_{i^*-1} \geq 1/\users$.  Next, we argue that $i^*$ is in $S$ except with small probability.  Notice that $\malenc(\userkeyvec, i^*)$ and $\malenc(\userkeyvec, i^*-1)$ differ only in the message encrypted under key $\userkeyi{i^*}$, so if $i^* \not\in S$, this key is unknown to the pirate decoder.  The security of the encryption scheme (made precise in Definition~\ref{def:encscheme}) guarantees that if $\userkeyi{i^*}$ is unknown to an efficient pirate, then we can replace $\tracequeries$ uses of $\enc(\userkeyi{i^*}, 1)$ with $\enc(\userkeyi{i^*}, 0)$, and this change will only affect the success probability of the pirate by $o(1/\users)$.  But after we make this replacement, $\malenc(\userkeyvec, i^*)$ and $\malenc(\userkeyvec, i^*-1)$ are (perfectly, information-theoretically) indistinguishable to the pirate.  Since the sequence of indices $i_1, \dots, i_{\tracequeries}$ is random, the pirate has no information about which elements $i_j$ are $i^*$ and which are $i^*-1$.  Thus, if the pirate wants to make $P_{i^*}$ larger than $P_{i^*-1}$, for some $i^* \not\in S$, she can do no better than to ``guess''.  If we take $s = \tilde{O}(\users^2)$, and apply a Chernoff bound, it turns out that for every $i \not\in S$, $P_{i} - P_{i-1} = o(1/\users)$.  This conclusion also holds after we take into account the security loss of the encryption scheme, which is $o(1/\users)$.  Thus, the scheme we described is a secure traitor-tracing scheme in the sense of Definition~\ref{def:ttscheme}.

In arguing that the scheme is secure, we used the fact that $P_{0} = 0$ and $P_{\users} = 1$ \emph{no matter what other queries are made to the pirate}.  In many applications, this assumption would not be reasonable.  However, when the pirate is derived from an accurate sanitizer, this condition will be satisfied.

For this scheme, the tracer makes $(n+1)s = \tilde{O}(\users^3)$ queries.
Before proceeding, we will explain how to reduce the number of queries from $\tilde{O}(\users^3)$ to $\tilde{O}(\users^2)$.  The high-level argument that the scheme is secure used two facts:
\begin{enumerate}
\item By the availability of the pirate decoder, if every user in $S$ would decrypt a ciphertext $\ctext$ to $b$, then the pirate decrypts $\ctext$ to $b$ (in the above, $P_0 = 0, P_\users = 1$).
\item Because of the encryption, a pirate decoder without user $i$'s key ``doesn't know'' how user $i$ would decrypt each ciphertext.
\end{enumerate}

Systems leveraging these two properties to identify a colluding user are called \emph{fingerprinting codes}~\cite{BonehSh98}, and have been studied extensively.  In fact, the tracing algorithm we described is identical to the tracing algorithm we define in Section~\ref{sec:thescheme}, but instantiated with the fingerprinting code of Boneh and Shaw~\cite{BonehSh98}, which has length $\tilde{O}(\users^3)$.  Tardos~\cite{Tardos08} constructed shorter fingerprinting codes, with length $\tilde{O}(\users^2)$, which we use to reduce the number of queries to trace.
 
 Next we define the precise security requirement we need out of the underlying encryption scheme, and then we will give a formal definition of fingerprinting codes.
 \subsection{Encryption Schemes} \label{sec:weakencryption}
We will build our traitor-tracing scheme from a suitable encryption scheme.  An encryption scheme is tuple of efficient algorithms $(\gen, \enc, \dec)$.  All the algorithms may be randomized except for $\dec$.  The scheme has the following syntactic properties:
\begin{itemize}
\item The algorithm $\gen$ takes a security parameter $\param$, runs in time $\poly(\param)$, and returns a private key $\userkey \in \bits^\param$.  Formally $\userkey \getsr \gen(1^\param)$.
\item The algorithm $\enc$ takes a private key and a message bit $b \in \bits$, runs in time $\poly(\param)$, and generates a ciphertext $\ctext \in \ctexts = \ctexts^{(\param)}$.  Formally, $\ctext \getsr \enc(\userkey, b)$.
\item The algorithm $\dec$ takes a private key $\userkey \in \bits^{\param}$ and a ciphertext $\ctext \in \ctexts^{(\param)}$, runs in time $\poly(\param)$, and returns a message bit $\decbit$.
\end{itemize}
First we define (perfectly) correct decryption\footnote{It would not substantially affect our results if $\dec$ were allowed to fail with negligible probability, however we will assume perfect correctness for ease of presentation.}
\begin{definition}[Correctness]  An encryption scheme $(\gen, \enc, \dec)$ is \emph{(perfectly) correct} if for every $b \in \bits$, and every $\param \in \N$,
$$
\Prob{\userkey \getsr \gen(1^\param)}{\dec(\userkey, \enc(\userkey, b)) = b} = 1.
$$
\end{definition}

We require that our schemes have the following $\encqueries$-message security property.

\begin{definition}[Security of Encryption] \label{def:encscheme}
Let $\encsecurity \from \N \to [0,1]$ and $\encqueries \from \N \to \N, \enctime \from \N \times \N \to \N$ be functions.  An encryption scheme $\encscheme = (\gen, \enc, \dec)$ is \emph{$(\encsecurity, \encqueries, \enctime)$-secure} if for every $\enctime(\param, \encqueries(\param))$-time algorithm $\encadv$ and every $b = (b_{1}, \dots, b_{\encqueries}), b' = (b'_{1}, \dots, b'_{\encqueries}) \in \bits$ (for $\encqueries = \encqueries(\param)$),
\begin{align*}
\Bigg| &\Prob{\userkey \getsr \gen(1^\param)}{\encadv(\enc(\userkey, b_1), \dots, \enc(\userkey, b_{\encqueries})) = 1} \\
&-  \Prob{\userkey \getsr \gen(1^\param) }{\encadv(\enc(\userkey, b'_1), \dots, \enc(\userkey, b'_{\encqueries})) = 1} \Bigg| 
\leq \encsecurity(\param).
\end{align*}
\end{definition}

Notice that we do not require $\encscheme$ to be secure against adversaries that are given $\enc(\userkey, \cdot)$ as an oracle.  That is, we do not require CPA security.

\begin{definition}[Encryption Scheme]
We say that a tuple of algorithms $\encscheme = (\gen, \enc, \dec)$ is an $(\encsecurity, \encqueries, \enctime)$-encryption scheme if it satisfies correctness and $(\encsecurity, \encqueries, \enctime)$-security.
\end{definition}
 
 \subsection{Fingerprinting Codes} \label{sec:fpcodes}
As we alluded to above, our tracing algorithm will use a \emph{fingerprinting code}, introduced by Boneh and Shaw~\cite{BonehSh98}.  A fingerprinting code is a pair of efficient (possibly randomized) algorithms $(\fpgen, \fptrace)$ with the following syntax.
\begin{itemize}
\item The algorithm $\fpgen$ takes a number of users $\users$ as input and outputs a codebook of $\users$ codewords of length $\wordlength = \wordlength(\users)$, $\wordbook = (\wordi{1}, \dots, \wordi{\users}) \in \bits^{\wordlength}$.  Formally $\wordbook \getsr \fpgen(1^\users)$.  We will think of $\wordbook \in \bits^{\users \times \wordlength}$ as a matrix with each row containing a codeword.
\item The algorithm $\fptrace$ takes an $\users$-user codebook $\wordbook$ and a word $\fword \in \bits^{\wordlength}$ and returns an index $i \in [\users]$.  Formally, $i = \fptrace(\wordbook, \fword)$.
\end{itemize}

Given a non-empty subset $S \subseteq [\users]$ and a set of codewords $\wordbook_{S} = (\wordi{i})_{i \in S} \in \bits^{\wordlength}$, we define the set of \emph{feasible codewords} to be
$$\feasible(\wordbook_{S}) = \set{\fword \in \bits^{\wordlength} \mid \forall j \in [\wordlength] \, \exists i \in S \; \fwordj{j} = \wordij{i}{j}}.
$$ Informally, if all users in $S$ have a $0$ (resp.~$1$) in the $j$-th symbol of their codeword, then they must produce a word with $0$ (resp.~$1$) as the $j$-th symbol.  We also define the $\emph{critical positions}$ to be the set of indices for which this constraint is binding.
That is, $$\mathrm{Crit}(\wordbook_{S}) = \set{j \in [\wordlength] \mid \forall i, i' \in S \; \wordij{i}{j} = \wordij{i'}{j}}.$$

The security of a fingerprinting code asserts that an adversary who is given a subset $\wordbook_{S}$ of the codewords should not be able to produce an element of $\feasible(\wordbook_{S})$ that does not trace to a user $i \in S$.  More formally,
\begin{definition}[Secure Fingerprinting Code] \label{def:fpscheme}
Let $\fpsecurity\from \N \to [0,1]$ and $\wordlength\from \N \to \N$ be functions.  A pair of algorithms $(\fpgen, \fptrace)$ is an \emph{$(\fpsecurity, \wordlength)$-fingerprinting code} if $\fpgen(1^\users)$ outputs a codebook $\wordbook \in \bits^{\users \times \wordlength(\users)}$, and furthermore, for every (possibly inefficient) algorithm $\fpadv$, and every non-empty $S \subseteq [\users]$,
$$
\Prob{\wordbook \getsr \fpgen(1^{\users})}{\fpadv(\wordbook_{S}) \in \feasible(\wordbook_{S}) \land \fptrace(\wordbook, \fpadv(\wordbook_{S})) \not\in S} \leq \fpsecurity(\users)
$$
where the two executions of $\fpadv$ are understood to be the same.
\end{definition}

Tardos~\cite{Tardos08} gave a construction of fingerprinting codes of essentially optimal length, improving on the original construction of Boneh and Shaw~\cite{BonehSh98}.
\begin{theorem} [\cite{Tardos08}] \label{thm:shortfpcodes}
For every function $\fpsecurity \from \N \to [0,1]$, there exists an $\left(\fpsecurity, O(\users^2 \log(\users/\fpsecurity))\right)$-fingerprinting code.  In particular, there exists a \linebreak $\left(o(1/\users^2), O(\users^2 \log \users)\right)$-fingerprinting code.
\end{theorem}

\subsection{The Traitor-Tracing Scheme}\label{sec:thescheme}
 We are now ready to state the construction more formally.  The key generation, encryption, and decryption algorithms are as we described in the sketch (Section~\ref{sec:sketch}), and stated below.
 
\begin{algorithm}
\caption{The algorithms $(\ttgen, \ttenc, \ttdec)$ for $\ttscheme$.  }
\begin{algorithmic}
\STATE{Let an encryption $\encscheme = (\gen, \enc, \dec)$ and a function $\users \from \N \to \N$ be parameters of the scheme.  Assume that $\users(\param) \leq 2^{\param/2}$ for every $\param \in \N$}
\STATE{}
\STATE{$\ttgen(1^\param)\colon$}
\STATE{For every user $i = 1, \dots, \users(\param)$, let $\enckeyi{i} \getsr \gen(1^{\param/2})$}
\STATE{Let $\userkeyi{i} = (\enckeyi{i}, i)$ (padded with zeros to have length exactly $\param$).}
\STATE{Output $\userkeyvec = (\userkeyi{1}, \dots, \userkeyi{\users})$}
\STATE{(We will sometimes use $\userkeyi{i}$ and $\enckeyi{i}$ interchangeably)}
\STATE{}
\STATE{$\ttenc(\userkeyi{1}, \dots, \userkeyi{\users}, b)\colon$}
\STATE{For every user $i$, let $\ctexti{i} \getsr \enc(\userkeyi{i}, b)$}
\STATE{Output $\ctext = (\ctexti{1}, \dots, \ctexti{\users})$}
\STATE{}
\STATE{$\ttdec(\userkeyi{i}, \ctext)\colon$}
\STATE{Output $\decbit = \dec(\userkeyi{i}, \ctexti{i})$}
\end{algorithmic}
\end{algorithm} 

\begin{algorithm}
\caption{The algorithm $\trace$ for $\ttscheme$}
\begin{algorithmic}
\STATE{The tracing algorithm for $\ttscheme$ and the subroutine $\malenc$.  Let a length $\wordlength = \wordlength(\users)$ fingerprinting code $\fpscheme = (\fpgen, \fptrace)$ be a parameter of the scheme and let $\encscheme = (\gen, \enc, \dec)$ be the encryption scheme used above.}
\STATE{}
\STATE{$\malenc( \userkeyi{1}, \dots, \userkeyi{\users}, W)$:}
\STATE{Let $\users \times \queries$ be the dimensions of $W$}
\STATE{For every $i \in [\users], j \in [\queries]$, let $\ctextij{i}{j} \getsr \enc(\userkeyi{i}, W_{i,j})$}
\STATE{For every $j \in [\queries]$, let $\ctextj{j} = (\ctextij{1}{j}, \dots, \ctextij{\users}{j})$}
\STATE{Output $\ctextj{1}, \dots, \ctextj{\queries}$}
\STATE{(Notice that $\dec(\userkeyi{i}, \ctextij{i}{j}) = W_{i,j}$)}
\STATE{}
\STATE{$\mathsf{Trace}_{\mathsf{TT}}^{\pirate}(\userkeyvec)$:}
\STATE{Let $\users$ be the number of user keys and $\wordlength = \wordlength(\users)$}
\STATE{Let $\wordbook \getsr \fpgen(1^\users)$}
\STATE{Let $\decbitj{1}, \dots, \decbitj{\wordlength} \getsr \pirate(\malenc(\userkeyvec, \wordbook))$ and let $\fword = \decbitj{1} \| \dots \| \decbitj{\wordlength}$}
\STATE{Output $i \getsr \fptrace(\wordbook, \fword)$}
\end{algorithmic}
\end{algorithm}

\subsection{Security of $\ttscheme$} \label{sec:weakenctott}

In this section we will prove that out construction of $\ttscheme = (\ttgen, \ttenc, \ttdec, \trace)$ is an $(\users, \wordlength(\users))$-secure traitor-tracing scheme.  It can be verified from the specification of the scheme that it has the desired syntactic properties, that it generates $\users(\param)$ user keys, and that the tracing algorithm makes $\wordlength(\users(\param))$ non-adaptive queries to its oracle.

Now we show how an available pirate decoder for this scheme can be traced.  As in the sketch (Section~\ref{sec:sketch}), we want to generate a set of ciphertexts that different users decrypt in different ways.  Specifically, given a fingerprinting code $W \in \bits^{\users \times \wordlength}$ (represented as a matrix with $\wordi{i}$ in the $i$-th row), we want to generate a set of ciphertexts $\ctextj{1}, \dots, \ctextj{\wordlength}$, such that user $i$, if she decrypts as intended using $\ttdec(\userkeyi{i}, \cdot)$, will decrypt $\ctextj{j}$ to $\wordij{i}{j}$.  That is, $\ttdec(\userkeyi{i}, \ctextj{j}) = \wordij{i}{j}$.  $\trace$ will query the pirate decoder on these ciphertexts, treat these responses as a word $\fword$, run the tracing algorithm for the fingerprinting code on $\fword$, and use the output of $\fptrace$ as its own output.  

If $\pirate$ is available, its output will be a feasible codeword for $\wordbook_{S}$.  To see this, recall that if every user $i \in S$ decrypts $\ctextj{j}$ to the same bit, then an available pirate decoder $\pirate(\userkeyvec_{S}, \cdot)$, decrypts $\ctextj{j}$ to that bit.  However, the critical positions of $\wordbook_{S}$ are exactly those for which every user $i \in S$ has the same symbol in position $j$.  Thus, the codeword returned by the pirate is feasible, and the fingerprinting code's tracing algorithm can identify a user in $S$.

The catch in this argument is that $\malenc$ takes all of $\wordbook$ as input, however an attacker for the fingerprinting code is only allowed to see $\wordbook_{S}$, and thus cannot simulate $\malenc$ in a security reduction.  However, if $\pirate$ only has keys $\userkeyvec_{S}$, and $i \not\in S$, then an efficient $\pirate$ cannot decrypt the $i$-th component of a ciphertext $\ctext = (\ctexti{1}, \dots, \ctexti{\users})$.  But these are the only components that depend on $\wordi{i}$.  So $\wordi{i}$ is computationally hidden from $\pirate$ anyway, and we could replace that codeword with a string of zeros without significantly affecting the success probability of $\pirate$.  Formalizing this intuition will yield a valid attacker for the fingerprinting code, and obtain a contradiction.

 \begin{theorem}[From Encryption to Traitor-Tracing] \label{thm:weakenctott}
Let $\encscheme$ be any \linebreak $(\encsecurity, \encqueries, \enctime)$-secure encryption scheme, and $\fpscheme$ be a $(\fpsecurity, \wordlength)$-fingerprinting code, $\fpscheme$.  Let $\users, \tracequeries \from \N \to \N$ be any functions such that for every $\param \in \N$, $\users(\param) \leq 2^{\param/2}$ and
\begin{enumerate}
\item the encryption scheme and fingerprinting code have sufficiently strong security,
$$
  \users(\param) \cdot \encsecurity(\param) +\fpsecurity(\users(\param)) = o\left( \frac{1}{\users(\param)^2} \right),
 $$
\item  the encryption scheme is secure for sufficiently many queries,
$$
 \encqueries(\param) \geq \tracequeries(\param) = \wordlength(\users(\param)),
 $$
 \item the encryption scheme is secure against adversaries whose running time is as long as the pirate decoder's, for every $a > 0$,
 $$
\enctime(\param/2, \tracequeries(\param)) \geq (\param + \users(\param) + \tracequeries(\param))^a.
 $$
 \end{enumerate}
 Then $\ttscheme$ instantiated with $\encscheme$ and $\fpscheme$ is an  $(\users, \tracequeries)$-traitor-tracing scheme.
 \end{theorem}
 \begin{proof}
Suppose there exists a $\poly(\param, \users(\param), \tracequeries(\param))$-time pirate decoder $\pirate$ that violates the security of $\ttscheme$.  That is, for every $\param \in \N$, there exists $S = S(\param) \subseteq [\users(\param)]$, $|S| \geq \users(\param) - 1$, such that
\begin{equation*}
\Prob{\userkeyvec \getsr \ttgen(1^{\param})}{\mathsf{Trace}_{\mathsf{TT}}^{\pirate(\userkeyvec_{S(\param)}, \cdot)}(\userkeyvec) \not\in S} = \Omega\left(\frac{1}{\users(\param)}\right)
\end{equation*}
where the probability is also taken over the coins of $\pirate$ and $\trace$.  Since there are only $\users(\param)$ such sets, for a randomly chosen $i \getsr [\users(\param)]$, we have
\begin{equation*}
\Prob{\userkeyvec \getsr \ttgen(1^{\param}) \atop i \getsr [\users(\param)]}{\mathsf{Trace}_{\mathsf{TT}}^{\pirate(\userkeyvec_{S_{-i}}, \cdot)}(\userkeyvec) \not\in S} = \Omega\left(\frac{1}{\users(\param)^2}\right).
\end{equation*}
Both of these probabilities are also taken over the coins of $\pirate$ and $\trace$.  We will show that such a pirate decoder must either violate the security of the encryption scheme or violate the security of the fingerprinting code. 

Given a matrix $\wordbook \in \bits^{(\users ) \times \wordlength(\users)}$, we define $\wordbook_{-i} \in \bits^{(\users - 1) \times \wordlength}$ to be $\wordbook$ with the $i$-th codeword removed and $\compbook_{-i} \in \bits^{\users \times \wordlength(\users)}$ to be $\wordbook$ with the $i$-th codeword replaced with $\vec{0}^{\wordlength(\users)}$.  We also use $S_{-i}$ as a shorthand for $[\users] \setminus \set{i}$

Consider the following algorithm $\fpadv^\pirate$
\begin{algorithm}
\caption{The fingerprinting security adversary.}
\begin{algorithmic}
\STATE{$\fpadv^{\pirate}(S_{-i}, \wordbook_{-i})\colon$}
\STATE{Let $\users$ be the number of users for the fingerprinting code and $\param$ be such that $\users(\param) = \users$}
\STATE{Generate keys $\userkeyvec \getsr \ttgen(1^\param)$ and ciphertexts $(\ctextj{1}, \dots, \ctextj{\wordlength}) \getsr \malenc(\userkeyvec, \compbook_{-i})$}
\STATE{Output $\fword = (\decbitj{1}, \dots, \decbitj{\wordlength}) \getsr \pirate(\userkeyvec_{-i}, \ctextj{1}, \dots, \ctextj{\wordlength})$}
\STATE{(Note that $\compbook_{-i}$ is just $\wordbook_{-i}$ with a row of zeros added, so the attacker is well-defined.)}
\end{algorithmic}
\end{algorithm}

Since the fingerprinting code is secure, for a randomly chosen $i \getsr [\users]$ (in fact, for every $i \in [\users]$),
\begin{equation} \label{eq:fpsuccess}
\Prob{\wordbook \getsr \fpgen(1^\users) \atop i \getsr [\users]}{\fpadv^{\pirate}(S_{-i}, \wordbook_{-i}) \in \feasible(\wordbook_{-i}) \land \fptrace(\wordbook, \fpadv^{\pirate}(S_{-i}, \wordbook_{-i})) = i} \leq \fpsecurity(\users)
\end{equation}
This condition could hold simply because $\fpadv$ outputs an infeasible codeword with high probability, not because we are successfully tracing a user in $S$.  The next claim states that if $\pirate$ is an available pirate decoder, then this is not the case.
\begin{claim} \label{clm:piratetofeasible}
Let $\tracequeries = \tracequeries(\param) = \wordlength(\users(\param))$ for every $\param \in \N$.  If $\pirate$ is a $\tracequeries$-available pirate decoder, then for every $\param \in \N$, every $i \in [\users(\param)]$, and every $\wordbook \in \bits^{\users \times \wordlength(\users)}$ (for $\users = \users(\param)$)
\begin{equation*}
\prob{\fpadv^{\pirate}(S_{-i}, \wordbook_{-i}) \not\in \feasible(\wordbook_{-i})} = o\left( \frac{1}{\users(\param)^2} \right)
\end{equation*}
\end{claim}
\begin{proof} [Proof of Claim~\ref{clm:piratetofeasible}]
If $\pirate$ is $\tracequeries$-useful, then, by definition, for every $\userkeyvec = (\userkeyi{1}, \dots, \userkeyi{\users})$, every $i \subseteq [\users]$, and every $\ctextseq$, if every user $i' \neq i$ decrypts some $\ctextj{j}$ to the same bit $b_j$, then so does $\pirate(\userkeyvec_{-i}, \cdot)$ (with high probability).  That is,
for $\decbitseq \getsr \pirate(\userkeyvec_{-i}, \ctextseq)$,
\begin{equation} \label{eq:restateuseful0}
\prob{\exists  j \in [\tracequeries], b \in \bits \left( \left( \forall i' \neq i, \ttdec(\userkeyi{i'}, \ctextj{j}) = b\right) \land \left( \decbitj{j} \neq b \right) \right)} = o\left(\frac{1}{\users(\param)^2}\right)
\end{equation}
Consider any critical position $j \in \crit(\wordbook_{-i})$.  These are the positions for which every user $i' \neq i$ has the same bit $\wordij{i'}{j} = b_j$.  It's easy to see from the definition of $\malenc$ (and the correctness of $\encscheme$) that if $\ctextseq \getsr \malenc(\userkeyvec, \compbook_{-i})$ then every user $i' \neq i$ will decrypt $\ctextj{j}$ to $b_j$.  Thus, with probability close to $1$, for every critical position $j$, the $j$-th output of $\pirate(\userkeyvec_{-i}, \ctextseq)$ will be equal to $b_j$, which implies $\fword = (\decbitj{1}, \dots, \decbitj{\wordlength})$ is feasible.
\end{proof}

Since $\pirate$ outputs feasible codewords with high probability, we obtain
\begin{equation} \label{eq:almosttracesuccess}
\Prob{\wordbook \getsr \fpgen(1^\users) \atop i \getsr [\users]}{\fptrace(\wordbook, \fpadv^{\pirate}(S_{-i}, \wordbook_{-i})) = i}  \leq \fpsecurity(\users(\param)) + o\left(\frac{1}{\users(\param)^2}\right)
\end{equation}
by combining the previous claim with~\eqref{eq:fpsuccess}.

There are only two differences between the success of the pirate decoder in fooling $\trace$ and the success of the fingerprinting adversary in fooling $\fptrace$ (in the experiment described in~\eqref{eq:almosttracesuccess}):  The first is that in the traitor-tracing security condition, $\pirate$ is given $\userkeyvec_{-i}$ for a fixed $i \in [\users]$, whereas the fingerprinting adversary is given $\wordbook_{-i}$ for a random $i \getsr [\users]$.  This difference only affects the error by a factor of $\users$.  That is, for every $i \in [\users]$
\begin{equation*}
\prob{\mathsf{Trace}_{\mathsf{TT}}^{\pirate(\userkeyvec_{-i}, \cdot)}(\userkeyvec) = i}
\leq \users \cdot \Prob{i \getsr [\users]}{\mathsf{Trace}_{\mathsf{TT}}^{\pirate(\userkeyvec_{-i}, \cdot)}(\userkeyvec) = i}
\end{equation*}
The second difference is that in $\trace$, the ciphertexts given to the pirate are generated by $\malenc(\userkeyvec, \wordbook)$ whereas in $\fpadv$ the ciphertexts are generated by \linebreak $\malenc(\userkeyvec, \compbook_{-i})$.  But these ciphertexts only differ in the $i$-th component, and $\userkeyi{i}$ is unknown to $\pirate$, so this does not affect the behavior of the pirate decoder significantly.  This fact is established in the following claim.
\begin{claim} \label{clm:distinguisher}
If $\encscheme$ is $(\encsecurity, \encqueries, \enctime)$-secure for $\encqueries, \enctime$ as in the statement of the Theorem, then for every $\poly(\param, \users(\param), \tracequeries(\param))$ pirate decoder $\pirate$,
\begin{align*}
\Bigg| &\Prob{\wordbook \getsr \fpgen(1^\users) \atop \userkeyvec \getsr \ttgen, i \getsr [\users]}{\fptrace(\wordbook, \pirate(\userkeyvec_{-i}, \malenc(\userkeyvec, \wordbook))) = i} \\
&-  \Prob{\wordbook \getsr \fpgen(1^\users) \atop \userkeyvec \getsr \ttgen, i \getsr [\users]}{\fptrace(\wordbook, \pirate(\userkeyvec_{-i}, \malenc(\userkeyvec, \compbook_{-i}))) = i} \Bigg| \leq \encsecurity(\param)
\end{align*}
\end{claim}
\begin{proof}[Proof of Claim~\ref{clm:distinguisher}]
Let $\encscheme = (\gen, \enc, \dec)$ be the encryption scheme.  The main observation required to prove the claim is that the two experiments we want to relate can both be simulated without $\userkeyi{i}$, given challenges for the encryption scheme (Definition~\ref{def:encscheme}).  Fix a codebook $W \getsr \fpgen(1^\users)$.  Now consider two distributions on ciphertexts (of $\encscheme$):  In either case, generate a random key $\userkeyi{i} \getsr \gen(1^\param)$
\begin{itemize} 
\item In the first case $\ctextij{i}{1} \getsr \enc(\userkeyi{i}, \wordij{i}{1}), \dots, \ctextij{i}{\wordlength} \getsr \enc(\userkeyi{i}, \wordij{i}{\wordlength})$
\item In the second case $\userkeyi{i} \getsr \gen(1^\param)$ and $\ctextij{i}{1} \getsr \enc(\userkeyi{i}, 0), \dots, \ctextij{i}{\wordlength} \getsr \enc(\userkeyi{i}, 0)$
\end{itemize}
Suppose we receive a set of $\wordlength$ ciphertexts from one of these two distributions.  Note that $\ttgen$ chooses keys for each user independently, and $\malenc$ generates ciphertext components for each user independently.  So we can generate keys $\userkeyvec_{-i}$, and ciphertext components for users other than $i$ independently, and use the challenge ciphertexts in place of the ciphertext components for user $i$, without knowing $\userkeyi{i}$.  Suppose we simulate $\malenc(\userkeyvec, \wordbook)$ in this way.  Notice that if the challenge ciphertexts come from the first distribution, then simulated ciphertexts will be distributed exactly as in $\malenc(\userkeyvec, \wordbook)$, and if the challenge ciphertexts come from the second distribution, then the simulated ciphertexts will be distributed exactly as in $\malenc(\userkeyvec, \compbook_{-i})$.  But, if the claim were false, then we would have found an adversary for the encryption scheme that can distinguish between the two distributions with advantage greater than $\encsecurity(\param)$.  It is easy to see that if the pirate decoder is efficient, then so will the adversary for the encryption scheme (since $\fptrace, \gen, \enc$ are all assumed to be efficient.  We conclude that if the claim is false, then $\encadv$ violates the security of $\encscheme$.
\end{proof}

We now complete the proof of the theorem by combining Equation~\eqref{eq:almosttracesuccess} and Claim~\ref{clm:distinguisher}.

\end{proof}

\subsection{Decryption Function Family of $\ttscheme$} \label{sec:decfams}

Recall that the two goals of constructing a new traitor-tracing scheme were to trace stateful pirates and to reduce the complexity of decryption.  We addressed tracing of stateful pirates in the previous section, and now we turn to the complexity of decryption.  We do so by instantiating the traitor-tracing scheme with various encryption schemes and making two observations: 1) The type of encryption schemes we require are sufficiently weak that there already exist plausible candidates with a very simple decryption operation, and 2) Decryption for the traitor-tracing scheme is not much more complex than decryption for the underlying encryption scheme.  We summarize the second point with the following simple lemma.
\begin{lemma}[Decryption Function Family for $\ttscheme$] \label{lem:formofttdec}
Let $\ttscheme$ be as defined, with $\encscheme$ as its underlying encryption scheme.  Let $(\enckey, i) = \userkey \in \bits^{\param}$ and $\ctext = (\ctexti{1}, \dots, \ctexti{\users}) \in \ctexts^{(\param)}$ be any user key and ciphertext for $\ttscheme$.  Then 
\begin{equation*}
\mathsf{Dec}_{\mathsf{TT}, \ctext}(\userkey) = \mathsf{Dec}_{\mathsf{TT}, \ctext}(\enckey, i)  = \bigvee_{i' \in [\users]} \left(\mathbf{1}_{i'}(i) \land \dec_{\ctexti{i'}}(\overline{\userkey}) \right)
\end{equation*}
\end{lemma}
Here, the function $\mathbf{1}_{x}(y)$ takes the value $1$ if $y = x$ and $0$ otherwise.  The lemma follows directly from the construction of $\ttdec$.  Also note that the function $\mathbf{1}_{i'} \from \bits^{\lceil \log \users \rceil} \to \bits$ is just a conjunction of $\lceil \log \users \rceil$ bits (a single gate of fan-in $O(\log \users)$), and we need to compute $\users$ of these functions.  In addition to computing $\mathbf{1}_{i'}$ and $\dec_{\ctexti{i'}}$, there are $\users$ conjunctions and a single outer disjunction.  Thus we add an additional $\users + 1$ gates, compute decryption $n$ times, and increase the depth by $2$.  Hence, an intuitive summary of the lemma is that if $\dec$ can be implemented by circuits of size $s$ and depth $h$, $\ttdec$ can be implemented by circuits of size $\users \cdot (s + O(\log \users)) = \tilde{O}(\users s)$ and depth $h + 2$.  This summary will be precise enough to state our main results.

By combining Lemma~\ref{lem:formofttdec} with Theorem~\ref{thm:weakenctott}, we easily obtain the following corollary.
\begin{corollary}[One-way Functions Imply Traitor-Tracing w/ Poly-Time Decryption]\label{cor:owftott}
Let $\users = \users(\param)$ be any polynomial in $\param$.  Assuming the existence of (non-uniformly secure) one-way functions, there exists an $(\users, \tilde{O}(\users^2))$-secure traitor-tracing scheme with decryption function family $\decfns_{\ttdec, \param}$ consisting only of circuits of size $\poly(\param)$
\end{corollary}
\begin{proof}
The existence of one-way functions implies the existence of an encryption scheme $\encscheme$ that is $(1/\param^a, \param^a, \param^a)$-secure for every constant $a > 0$ and sufficiently large $\param$ with decryption function $\decfns_{\dec, \param}$ consisting only of circuits of size $t(\param) = \poly(\param)$ for every $\param \in \N$.  From Lemma~\ref{lem:formofttdec}, it is easy to see that if $\ttscheme$ uses $\encscheme$ as its encryption scheme, then $\decfns_{\ttdec, \param}$ consists only of circuits of size $\tilde{O}(\users(\param) t(\param/2)) = \poly(\param)$.
\end{proof}
Theorem~\ref{thm:main1} in the introduction follows by combining Theorem~\ref{thm:attack} with Corollary~\ref{cor:owftott}.

We will now consider the possibility of constructing a traitor-tracing scheme where the decryption functionality can be implemented by circuits of constant depth, and thus obtaining hardness results for generic sanitizers that are efficient for constant-depth queries (Theorem~\ref{thm:main2}).  First, we summarize our observation that the traitor-tracing scheme almost preserves the depth of the decryption function.
\begin{corollary}[Encryption with Constant-Depth Decryption Impies Traitor-Tracing w/ Constant-Depth Decryption] \label{cor:preservedepth}
Let $\users = \users(\param)$ be any polynomial in $\param$.  If there exists an encryption scheme, $(\gen, \enc, \dec)$, that is $(o(1/\users^2), \omega(\users^4), \users^a)$-secure for every $a > 0$ and has decryption family $\decfns_{\dec}^{(\param)}$ consisting of circuits of size $\poly(\param)$ and depth $h$, then there exists a $(\users, \tilde{O}(\users^2))$-secure traitor-tracing scheme with decryption function family $\decfns_{\ttdec}^{(\param)}$ consisting of circuits of size $\tilde{O}(\users) \cdot \poly(\param)$ and depth $h+2$.
\end{corollary}
The corollary is clear from Lemma~\ref{lem:formofttdec} and Theorem~\ref{thm:weakenctott}.

The corollary is not interesting without an encryption scheme that can be decrypted by constant-depth circuits.  However, we observe that such a scheme (meeting our relaxed security criteria) can be constructed from a sufficiently good \emph{local pseudorandom generator (PRG)}.  A recent result of Applebaum~\cite{Applebaum12} gave the first plausible candidate construction of a local PRG for the range of parameters we need, giving plausibility to the assumption that such PRGs (and, as we show, traitor-tracing schemes with constant-depth decryption) exist.  We note that local PRGs actually imply encryption schemes with local decryption, which is stronger than just constant-depth decryption.  Although it may be significantly easier to construct encryption schemes that only have constant-depth decryption, we are not aware of any other ways of constructing such a scheme.

\begin{definition}[Local Pseudorandom Generator]
An efficient algorithm \linebreak $\prg \from \bits^\param \to \bits^{\stretch(\param)}$ is a $\prgsecurity$-pseudorandom generator if for every $\poly(\stretch(\param))$-time adversary $\prgadv$
$$
\left| \prob{\prgadv(\prg(U_\param)) = 1}  - \prob{\prgadv(U_{\stretch(\param)}) = 1} \right| \leq \prgsecurity(\param)
$$
If, in addition, if each bit of the output depends only on some set of $L$ bits of the input, then $\prg$ is a $(\prgsecurity, L)$-local pseudorandom generator.
\end{definition}

It is a well known result in Cryptography that pseudorandom generators imply encryption schemes satisfying Definition~\ref{def:encscheme} (for certain ranges of parameters). We will use a particular construction whose decryption can be computed in constant-depth whenever the underlying PRG is locally-computable (or, more generally, computable by constant-depth circuits).
The construction is the standard ``computational one-time pad'', however we give a construction to verify that the decryption can be computed by constant-depth circuits.
\begin{algorithm}
\caption{An encryption scheme $\Pi_{\mathsf{LocalEnc}}$ that can be decrypted in constant depth.}
\begin{algorithmic}
\STATE{$\gen(1^\param)\colon$}
\STATE{Let $s \getsr \bits^\param$ and output $\userkey = s$}
\STATE{}
\STATE{$\enc(\userkey, b)\colon$}
\STATE{Let $r \getsr \set{1,2,\dots,\stretch(\param)}$ and output $\ctext = (r, \prg(\userkey)_{r} \oplus b)$}
\STATE{}
\STATE{$\dec(\userkey, \ctext)\colon$}
\STATE{Let $(r', b') = \ctext$ and output: $b = \prg(\userkey)_{r} \oplus b'$}
\end{algorithmic}
\end{algorithm}
\begin{lemma}[Local PRGs $\rightarrow$ Encryption] \label{lem:lprgtoenc}
If there exists a $(\prgsecurity(\param), L)$-local pseudorandom generator $\prg \from \bits^{\param} \to \bits^{\stretch(\param)}$, then there exists an $(\encsecurity = \prgsecurity + \encqueries^2 / \stretch, \encqueries)$-Secure Encryption Scheme $(\gen, \enc, \dec)$ with decryption function family $\decfns_{\dec, \param}$ consisting of circuits of size $\poly(\param)$ and depth $4$.
\end{lemma}

\begin{proof}
The security follows from standard arguments:  If we choose a random $s \getsr \bits^{\param}$, then $\prg(s)$ is indistinguishable from uniform up to error $\prgsecurity$.  If we generate $\encqueries$ encryptions with key $s$, and no two encryptions use the same choice of $r$, then the output is indistinguishable from encryptions using uniform random bits in place of $\prg(s)$. If we use uniform random bits in place of $\prg$, then the message is information-theoretically hidden.  The probability that no two encryptions out of $\encqueries$ use the same choice of $r$ is at most $\encqueries^2/\stretch$, so we lose this term in the security of the encryption scheme.

Let $\mathbf{1}_{i}(j)$ be the indicator variable for the condition $j=i$.  For every $\ctext = (r, b) \in \ctexts$, we can write
\begin{equation*}
\dec_{(r,b)}(s) = \bigvee_{i \in [\stretch(\param)]} \left( \mathbf{1}_{i}(r) \land \left(\prg_i(s) \oplus b\right) \right).
\end{equation*}
Observe that, since $\prg_i$ is a function of $L$ bits of the input, it can be computed by a size-$2^L$ DNF (depth-$2$ circuit), thus $\prg_i(s) \oplus b$ can be computed by a size $2^L + 1$, depth-$3$ circuit.  The indicator $\mathbf{1}_{i}$ can be computed by a conjunction of $ \lceil \log_2 \stretch(\param) \rceil$ bits, which is a size-$\lceil\log_2 \stretch(\param) \rceil$, depth-$1$ circuit.  The outer disjunction increases the depth by one level and the size by $1$.  Putting it all together, we have shown that $\dec_{r,b}(s)$ can be computed by depth-$4$ circuits of size $\tilde{O}(2^{L} \stretch(\param)) = \poly(\stretch(\param))$.
\end{proof}

Combining Corollary~\ref{cor:preservedepth} with Lemma~\ref{lem:lprgtoenc} easily yields the following corollary.
\begin{corollary} [Local Pseudorandom Generators Imply traitor-tracing w/ $AC^0$ Decryption] \label{cor:lenctott}
Let $\users = \users(\param)$ be any polynomial in $\param$.  Assuming the existence of a $(o(1/\users^2), \users^7, L)$-local pseudorandom generator for some constant $L \in \N$, there exists an $(\users, \tilde{O}(\users^2))$-secure traitor-tracing scheme with decryption function family $\decfns_{\ttdec, \param}$ consisting of circuits of size $\tilde{O}(\users)\cdot\poly(\param)$ and depth $6$.
\end{corollary}

Theorem~\ref{thm:main2} in the introduction follows by combining Theorem~\ref{thm:attack} with Corollary~\ref{cor:lenctott}.

\section*{Acknowledgements}
We thank Cynthia Dwork for suggesting that we look further at the connection between traitor-tracing and differential privacy.  We thank Salil Vadhan for helpful discussions about the connection between traitor-tracing and differential privacy, and about the presentation of this work.  We also thank Dan Boneh, Moritz Hardt, Hart Montgomery, Ananth Raghunathan, Aaron Roth, Guy Rothblum, and Thomas Steinke for helpful discussions.

\ifnum\sicompversion=1
	\newcommand{\etalchar}[1]{$^{#1}$}
	
\else
	\bibliographystyle{alpha}
	\bibliography{./privacyrefs}
\fi

\appendix

\section{Additional Related Work} \label{sec:conc}

In this appendix, we elaborate on the relationship between sanitizers, interactive sanitizers, and one-shot sanitizers, and on the relationship between our results and prior work on the complexity of differentially private sanitization.

\paragraph{The Relationship with~\cite{DworkNaVa12}}
Dwork, Naor, and Vadhan~\cite{DworkNaVa12} gave information theoretic lower bounds for \emph{stateless sanitizers}.  These are sanitizers that take $k$ queries as input, but whose answers to each query do not depend on the other $k-1$ input queries.    

Another interpretation of our results, which can be used to give an alternative proof of our results, is that we construct a family of queries for which ``keeping state doesn't help''.

They consider a game where an $n$-row database is chosen at random, and a random subset of $n-1$ of those rows is given to the attacker.  The attacker wants to violate privacy by recovering the $n$-th row.  To do so, the attacker chooses $\sim n^2$ queries (randomly, from a distribution that depends on the $n-1$ known rows) and requests answers to these queries.  Using these answers, they show that there is a particular way for the attacker to (randomly) guess the missing row, that will succeed with sufficient probability to constitute a privacy violation.  Their argument is in two steps: 1) The expected correlation between the answers given by a stateless sanitizer and the value of the queries on the missing row is significant.  2) A stateless sanitizer cannot give answers that are correlated with too many rows that are not in the database.  Combining these two steps shows that the attacker has a significant chance of identifying the $n$-th row from its correlation with the answers.

Typically, the intuition behind the analysis of traitor-tracing schemes follows roughly the same two steps:  1) There will be some correlation between the decryptions returned by the efficient pirate and the decryptions that would be returned by some member of the coalition (using only his own key).  2) There will not be significant correlation between the decryptions returned by the efficient pirate and the decryptions that would be returned by any user not a member of the coalition.  This is exactly the intuition we sketched in the simpler (sub-optimal) construction.  In the final construction, much of this argument is made ``inside'' the construction of the fingerprinting code.  If we ``unrolled'' the analysis of the fingerprinting code directly into our construction, we would make exactly the same arguments.  

\paragraph{Other Types of Sanitizers}
There are two other variants of sanitizers that have appeared in the literature.  The first, which we have already discussed, is a one-shot sanitizer.  The second, is an \emph{interactive sanitizer}.  This type of sanitizer is the same as the one we consider in this work, where the sanitizer is given a database $D$ and $k$ queries from a family $\cQ$, but the queries arrive one at a time, and may be chosen adaptively.  In this setting, we want the sanitizer to answer each query efficiently (in time polynomial in $d$, $n$, and $k$).  The Laplace Mechanism is, in fact, interactive, and a line of work initiated by Roth and Roughgarden~\cite{RothRo10, HardtRo10, GuptaRoUl12} showed how to interactively answer $2^{\tilde{O}(n)}$ queries in time $\poly(2^d, n)$ per query.

The three variants we've described satisfy some interesting relationships.  First, if we have an algorithm that runs in time $T$ and releases a summary that enables an analyst to answer any query in $\cQ$ in time $T$, then we also have an interactive sanitizer that runs in time $2T$ per query that answers any sequence of $k$ queries from $\cQ$.  Secondly, if we have an interactive sanitizer that answers up to $K$ queries from $\cQ$ in time $T$ per query, then we also have a non-interactive sanitizer that answers any $k \leq K$ queries from $\cQ$ in time $Tk$.  Thus, holding $\cQ$ fixed and assuming $k \gg n^2$, the problem we consider is the easiest form of private counting query release, and the lower bounds we prove imply lower bounds for the other variants.  

For the case of interactive sanitization, these lower bounds are new.  To our knowledge, prior to our work it was possible that there was an interactive sanitizer that ran in time $\poly(d,n)$ per query and answered nearly $2^n$ arbitrary counting queries, whereas our results show that there is no efficient interactive sanitizer for significantly more than $n^2$ queries.  On the other hand, for counting query release, our results only show that it is hard to release a particular family of queries $\cQ$ whose size is at least $2^n$.  For families of queries this large, the results of Dinur and Nissim~\cite{DinurNi03} already imply the impossibility of release, even by computationally unbounded algorithms.

Indeed, in the data release problem one cannot take $\cQ$ to include all efficiently computable counting queries.  Sanitizers (both interactive and non-interactive) are supposed to circumvent this problem by allowing the queries to be arbitrary, but only answering the $k$ queries that are needed.  However our results show that they can only circumvent the problem if we allow superpolynomial computation or we take $k \lesssim n^2$.

\end{document}